\documentclass[10pt,a4paper]{amsart}
\usepackage[margin=3.3cm]{geometry}
\usepackage{amssymb}
\usepackage[mathscr]{euscript}
\usepackage{enumerate}
\usepackage{xspace}
\usepackage{color}
\usepackage{enumerate}
\usepackage{enumitem}
\usepackage{amsthm}

\usepackage{hyperref}
\usepackage{natbib}
\DeclareMathAlphabet{\mathpzc}{OT1}{pzc}{m}{it}

\begin{document}

\theoremstyle{plain}
\newtheorem{theorem}{Theorem}[section]
\newtheorem{condition}{Condition}[section]
\newtheorem{lemma}{Lemma}[section]
\newtheorem{proposition}{Proposition}[section]
\newtheorem{corollary}{Corollary}[section]
\newtheorem{definition}{Definition}[section]
\newtheorem{Ass}{Assumption}[section]
\newcommand{\q}{Q}
\theoremstyle{definition}
\newtheorem{remark}{Remark}[section]
\newtheorem{example}{Example}[section]
\newtheorem{assumption}{Assumption}[section]
\newcommand{\notiz}{\textup} %{\relax}
\newtheorem{Case}{Case}[section]
\newcommand{\lebesgue}{\ensuremath{\lambda\!\!\lambda}}
\renewenvironment{proof}{{\parindent 0pt \it{ Proof:}}}{\mbox{}\hfill\mbox{$\Box\hspace{-0.5mm}$}\vskip 16pt}
\newenvironment{proofthm}[1]{{\parindent 0pt \it Proof of Theorem #1:}}{\mbox{}\hfill\mbox{$\Box\hspace{-0.5mm}$}\vskip 16pt}
\newenvironment{prooflemma}[1]{{\parindent 0pt \it Proof of Lemma #1:}}{\mbox{}\hfill\mbox{$\Box\hspace{-0.5mm}$}\vskip 16pt}
\newenvironment{proofcor}[1]{{\parindent 0pt \it Proof of Corollary #1:}}{\mbox{}\hfill\mbox{$\Box\hspace{-0.5mm}$}\vskip 16pt}
\newenvironment{proofprop}[1]{{\parindent 0pt \it Proof of Proposition #1:}}{\mbox{}\hfill\mbox{$\Box\hspace{-0.5mm}$}\vskip 16pt}

\newcommand{\pK}{\p^{\mathcal{K}}}
\newcommand{\Law}{\ensuremath{\mathop{\mathrm{Law}}}}
\newcommand{\loc}{{\mathrm{loc}}}
\newcommand{\Log}{\ensuremath{\mathop{\mathcal{L}\mathrm{og}}}}
\newcommand{\Meixner}{\ensuremath{\mathop{\mathrm{Meixner}}}}
\newcommand{\of}{[\hspace{-0.06cm}[}
\newcommand{\gs}{]\hspace{-0.06cm}]}
\newcommand{\A}{\mathbf{A}}
\newcommand{\B}{\mathbf{B}}

\let\MID\mid
\renewcommand{\mid}{|}

\let\SETMINUS\setminus
\renewcommand{\setminus}{\backslash}

\def\stackrelboth#1#2#3{\mathrel{\mathop{#2}\limits^{#1}_{#3}}}

\renewcommand{\theequation}{\thesection.\arabic{equation}}
\numberwithin{equation}{section}

\newcommand\llambda{{\mathchoice
      {\lambda\mkern-4.5mu{\raisebox{.4ex}{\scriptsize$\backslash$}}}
      {\lambda\mkern-4.83mu{\raisebox{.4ex}{\scriptsize$\backslash$}}}
      {\lambda\mkern-4.5mu{\raisebox{.2ex}{\footnotesize$\scriptscriptstyle\backslash$}}}
      {\lambda\mkern-5.0mu{\raisebox{.2ex}{\tiny$\scriptscriptstyle\backslash$}}}}}

\newcommand{\prozess}[1][L]{{\ensuremath{#1=(#1_t)_{0\le t\le T}}}\xspace}
\newcommand{\prazess}[1][L]{{\ensuremath{#1=(#1_t)_{0\le t\le T^*}}}\xspace}
\newcommand{\pH}{\p^{\mathcal{H}}}
\newcommand{\tr}{\operatorname{tr}}
\newcommand{\lijepoa}{{\mathcal{A}}}
\newcommand{\lijepob}{{\mathcal{B}}}
\newcommand{\lijepoc}{{\mathcal{C}}}
\newcommand{\lijepod}{{\mathcal{D}}}
\newcommand{\lijepoe}{{\mathcal{E}}}
\newcommand{\lijepof}{{\mathcal{F}}}
\newcommand{\lijepog}{{\mathcal{G}}}
\newcommand{\lijepok}{{\mathcal{K}}}
\newcommand{\lijepoo}{{\mathcal{O}}}
\newcommand{\lijepop}{{\mathcal{P}}}
\newcommand{\lijepoh}{{\mathcal{H}}}
\newcommand{\lijepom}{{\mathcal{M}}}
\newcommand{\lijepou}{{\mathcal{U}}}
\newcommand{\lijepov}{{\mathcal{V}}}
\newcommand{\lijepoy}{{\mathcal{Y}}}
\newcommand{\cF}{{\mathcal{F}}}
\newcommand{\cG}{{\mathcal{G}}}
\newcommand{\cH}{{\mathcal{H}}}
\newcommand{\cM}{{\mathcal{M}}}
\newcommand{\cD}{{\mathcal{D}}}
\newcommand{\bD}{{\mathbb{D}}}
\newcommand{\bF}{{\mathbb{F}}}
\newcommand{\bG}{{\mathbb{G}}}
\newcommand{\bH}{{\mathbb{H}}}
\newcommand{\dd}{\operatorname{d}\hspace{-0.055cm}}
\newcommand{\ddd}{\operatorname{d}}
\newcommand{\er}{{\mathbb{R}}}
\newcommand{\ce}{{\mathbb{C}}}
\newcommand{\erd}{{\mathbb{R}^{d}}}
\newcommand{\en}{{\mathbb{N}}}
\newcommand{\de}{{\mathrm{d}}}
\newcommand{\im}{{\mathrm{i}}}
\newcommand{\set}[1]{\ensuremath{\left\{#1\right\}}}
\newcommand{\indik}{{\mathbf{1}}}
\newcommand{\D}{{\mathbb{D}}}
\newcommand{\E}{{\mathbf{E}}}
\newcommand{\N}{{\mathbb{N}}}
\newcommand{\Q}{{\mathbb{Q}}}
\renewcommand{\P}{{\mathbb{P}}}
\newcommand{\ud}{\operatorname{d}\!}
\newcommand{\ii}{\operatorname{i}\kern -0.8pt}
\newcommand{\Var}{\operatorname{Var }\,}
\newcommand{\dt}{\operatorname{d}\!t}   %\newcommand{\dt}{\mbox{d$t$}}
\newcommand{\ds}{\operatorname{d}\!s}   %{\mbox{d$s$}}
\newcommand{\dy}{\operatorname{d}\!y }    %{\mbox{d$y$}}
\newcommand{\du}{\operatorname{d}\!u}  %{\mbox{d$u$}
\newcommand{\dv}{\operatorname{d}\!v}   %{\mbox{d$v$}}
\newcommand{\dx}{\operatorname{d}\!x}   %{\mbox{d$x$}}
\newcommand{\dq}{\operatorname{d}\!q}   %{\mbox{d$q$}}
\newcommand{\cadlag}{c\`adl\`ag }
\newcommand{\p}%{{\operatorname{P}}} %
{P}
\newcommand{\F}{\mathbf{F}}
\newcommand{\1}{\mathbf{1}}
\newcommand{\f}{\mathcal{F}^{\hspace{0.03cm}0}}
\newcommand{\BF}{B^{X, \F}(h)}
\newcommand{\CF}{C^{X, \F}}
\newcommand{\nuF}{\nu^{X, \F}}
\newcommand{\BG}{B^{X, \G}(h)}
\newcommand{\CG}{C^{X, \G}}
\newcommand{\nuG}{\nu^{X, \G}}
\newcommand{\BK}{B^{X, \mathbf{K}}(h)}
\newcommand{\CK}{C^{X, \mathbf{K}}}
\newcommand{\nuK}{\nu^{X, \mathbf{K}}}
\newcommand{\cf}{c^{X, \F}}
\newcommand{\G}{\mathfrak{G}}
\newcommand{\M}{\mathcal{M}^{\textup{sp}}}
\newcommand{\K}{\mathbb{K}}
\def\EM{\ensuremath{(\mathbb{EM})}\xspace}
\newcommand{\Cm}{\mathcal{C}_{b, i}^2(\mathbb{R})}
\newcommand{\Cb}{\mathcal{C}_b^2(\mathbb{R})}
\newcommand{\la}{\langle}
\newcommand{\ra}{\rangle}
\newcommand{\uX}{\underline{Z}}
\newcommand{\oX}{\overline{Z}}
\newcommand{\oRd}{|\mathbb{R}^d}
\newcommand{\Norml}[1]{%
{|}\kern-.25ex{|}\kern-.25ex{|}#1{|}\kern-.25ex{|}\kern-.25ex{|}}

\newcommand{\lle}{\langle\hspace{-0.085cm}\langle}
\newcommand{\rre}{\rangle\hspace{-0.085cm}\rangle}

\title[Deterministic Criteria for the Absence and Existence of Arbitrage]{Deterministic Criteria for the Absence and Existence of Arbitrage in Multi-Dimensional Diffusion Markets}

\author[D. Criens]{David Criens}
\address{D. Criens - Technical University of Munich, Department of Mathematics, Germany}
\email{david.criens@tum.de}
\thanks{D. Criens - Technical University of Munich, Department of Mathematics, Germany,  \texttt{david.criens@tum.de}.}
\keywords{Arbitrage, Equivalent (Local) Martingale Measure, Multi-Dimensional Time-Inhomogeneous Diffusion, Martingale Problem}

\subjclass[2010]{60G44, 60H10, 60J60, 91B70}

\date{\today}
\maketitle

\frenchspacing
\pagestyle{myheadings}

\begin{abstract}
	We derive deterministic criteria for the existence and non-existence of equivalent (local) martingale measures for financial markets driven by multi-dimensional time-inhomogeneous diffusions.
Our conditions can be used to construct financial markets in which the \emph{no unbounded profit with bounded risk} condition holds, while the classical \emph{no free lunch with vanishing risk} condition fails. 
\end{abstract}

\section{Introduction}
The question when a financial model is free of arbitrage is typically asked for each financial model individually. Our goal is to provide a systematic discussion for exponential models driven by multi-dimensional time-inhomogeneous diffusions.

We explain our setting in more detail. For a fixed finite time horizon \(T > 0\), let
\(X = (X_t)_{t \in [0, T]}\) be the coordinate process on the path-space \(\Omega\) of continuous functions \([0, T] \to \mathbb{R}^d\) and let \(\mathcal{F}\) be the \(\sigma\)-field generated by \(X\).
The real-world measure \(P\) of our financial market is a probability measure on \((\Omega, \mathcal{F})\) such that \(X\) is a diffusion parameterized by a drift coefficient \(b\) and a diffusion coefficient \(a\). Below, we will formally introduce \(P\) as a solution to a martingale problem (MP). 
Our financial market consists of \(m \leq d\) risky assets. Each of them is modeled as a discounted asset process \((S^i_t)_{t \in [0, T]}\) with dynamics
\begin{align}\label{eq: DPP}
\dd S^i_t = S^i_t  \langle e_i, \dd X_t\rangle
\end{align}
and deterministic initial value \(S_0^i > 0\).
Here, \(e_i\) is the \(i\)-th unit vector and \(\la \cdot, \cdot \ra\) is the Euclidean scalar product. Since \(m<d\) is possible, our setting includes incomplete markets and stochastic volatility models. 

For this financial market the classical concepts of no-arbitrage are the notion of \emph{no free lunch with vanishing risk (NFLVR)} as defined by \cite{DelbaenSchachermayer94, DelbaenSchachermayer98} and the notion \emph{no generalized arbitrage (NGA)} as defined by \cite{Cherny2007} and \cite{Yan97}.
The difference between (NFLVR) and (NGA) is captured by the concept of a \emph{financial bubble} in the sense of \cite{Cox2005}.
More precisely, a financial bubble exists if (NFLVR) holds while (NGA) fails. 
In diffusion markets it is well-known that (NFLVR) is equivalent to the existence of an \emph{equivalent local martingale measure (ELMM)}, see \cite{DelbaenSchachermayer94, DelbaenSchachermayer98}, and that (NGA) is equivalent to the existence of an \emph{equivalent martingale measure (EMM)}, see \cite{Cherny2007}. Consequently, a financial bubble exits if there is an ELMM but no EMM.

The no arbitrage condition used in the stochastic portfolio theory of \cite{fernholz2002stochastic} is \emph{no relative arbitrage (NRA)}.
In complete settings, (NRA) is equivalent to the existence of a \emph{strict martingale density (SMD)}, see \cite{fernholz2010}.

More recently, various types of weaker notions of no arbitrage are introduced and studied, see the article of \cite{doi:10.1142/S0219024915500053} for an overview.
One of these weak notions is the \emph{no unbounded profit with bounded risk (NUPBR)} condition. It is known that (NUPBR) is equivalent to the existence of a \emph{strict local martingale density (SLMD)}, see \cite{Choulli1996}.

In our diffusion market, the existence of a SLMD is equivalent to a mild drift condition.
On the foundation of this drift condition and a local Novikov condition, the contributions of this article are deterministic criteria for the existence and non-existence of E(L)MMs and SMDs. Our criteria are typically weaker than classical Novikov-type conditions and easy to verify.
In particular, our results can be used to construct financial markets which include financial bubbles and financial markets in which (NUPBR) holds, while (NFLVR) fails.
Latter receive recently more and more attention, see, for an example in a diffusion setting, the discussion of hedging under (NUPBR) given by \cite{MAFI:MAFI502}.

Let us comment on existing literature.
The absence of arbitrage for a class of diffusion models was for instance studied by \cite{MAFI:MAFI530}, \cite{Delbaen2002} and \cite{Mijatovic2012}.
\cite{MAFI:MAFI530} works in a market driven by a continuous It\^o process. He shows that (NFLVR) is determined by the equivalence of a probability measure to the Wiener measure. 
The spirit of the work of \cite{Delbaen2002} and \cite{Mijatovic2012} is very similar and close to ours. 
Their goal is to derive deterministic criteria which are suitable for applications.
However, they work in a one-dimensional setting, while we are particularly interested in multi-dimensional cases, which differ significantly from the one-dimensional one. We illustrate this by constructing a type of financial market which allows for arbitrage opportunities only if it features more than three risky assets. 
This is particularly surprising when noting that the sources of risk in this market have the same dimension. 
Moreover, our setting includes cases where the MP corresponding to the real-world measure may have more than one solution and cases with more than one candidate for an E(L)MM. In the setting of \cite{Delbaen2002}, respectively the setting of \cite{Mijatovic2012}, the real-world measure is unique and a single candidate for an E(L)MM exists. We give an example where not all candidates are E(L)MMs and an example where the existence of an E(L)MM may depend on the choice of the real-world measure, see Examples \ref{exp: Gir} and \ref{example: dependence P} below.

This article is structured as follows. In Section \ref{section 2} we recall the definitions of MPs and E(L)MMs. In Section \ref{sec: ex} respectively Section \ref{sec: nonex} we give deterministic conditions for the existence respectively non-existence of E(L)MMs and SMDs. In particular, in Section \ref{sec: SLMD} we discuss a drift condition which is equivalent to the existence of a SLMD.
We present an application in Section \ref{section 4}.
Finally, in Section \ref{sec:LMU} we comment on the results of \cite{Mijatovic2012} and \cite{MAFI:MAFI530}. For a comparison of the results of \cite{Mijatovic2012} and \cite{Delbaen2002} we refer to Remark 3.2 in \cite{Mijatovic2012}.

To keep this article self-contained, we have attached appendices which include almost all classical results from stochastic analysis which are used in our proofs. 

\section{The Financial Market}\label{section 2}
We fix a finite time horizon \(T > 0\).
Let \(\Omega\) be the space of continuous function \([0,T] \to \mathbb{R}^d\) for a fixed dimension \(d \in \mathbb{N}\).
The coordinate process \(X = (X_t)_{t \in [0, T]}\) on \(\Omega\) is given by \(X_t(\omega) = \omega(t)\). 
Moreover, we define by \(\mathcal{F}\triangleq \sigma(X_t, t \in [0, T])\) a \(\sigma\)-field on \(\Omega\) and equip \(\Omega\) with the uniform topology. It is well-known that \(\mathcal{F}\) is the Borel \(\sigma\)-field on \(\Omega\).

We set \(\mathcal{F}_t \triangleq \sigma(X_s, s \in [0, t])\) for \(t \in [0, T]\) and \(\F \triangleq (\mathcal{F}_{t+})_{t \in [0, T]}\), where \(\mathcal{F}_{t+} \triangleq \bigcap_{s \in (t, T]} \mathcal{F}_s\). %In this article we do not work under the usual hypothesis. We comment on this in Appendix \ref{App: comments mb}.
Our financial model will be parameterized by a triplet \((b, a, x_0)\), where 
\begin{enumerate}
	\item[(i)] \(b \colon [0, T] \times \mathbb{R}^d \to \mathbb{R}^d\) is a Borel function.
	\item[(ii)] \(a \colon [0, T] \times \mathbb{R}^d \to \mathbb{S}^d\) is a Borel function and \(\mathbb{S}^d\) is the set of all symmetric non-negative definite \(d\times d\) matrices.
	\item[(iii)] \(x_0 \in \mathbb{R}^d\).
\end{enumerate}
The parameter \(b\) corresponds to the drift, the parameter \(a\) corresponds to the diffusion part and \(x_0\) is the initial value.

The probabilistic concept underlying our model is a local version of the martingale problem as introduced by \cite{SV}, see also \cite{J79}. Let \(C^2(\mathbb{R}^d)\) be the set of all twice continuously differentiable functions \(\mathbb{R}^d \to \mathbb{R}\). For \(f \in C^2(\mathbb{R}^d)\) we set 
\begin{align}
\mathcal{K}^{b, a}_t f (x) \triangleq \la b(t, x), \nabla f (x)\ra + \frac{1}{2} \textup{trace } (a(t, x) \nabla^2 f(x)), 
\end{align}
where \(\nabla f\) denotes the gradient of \(f\) and \(\nabla^2 f\) denotes the Hessian matrix of \(f\).
\begin{definition}\label{def: MP}
	We call a probability measure \(P\) on \((\Omega, \mathcal{F})\) a solution to the martingale problem (MP) \((b, a, x_0)\), if 
	\(
	P  (X_0 = x_0) = 1
	\)
	and for all \(f \in C^2(\mathbb{R}^d)\) the process
	\begin{equation}\label{Mf}
	\begin{split}
	M^f_{\cdot, x_0} \triangleq f(&X_\cdot) - f(x_0) - \int_0^\cdot \mathcal{K}^{b, a}_s f(X_s) \dd s
%	\\&- \int_0^{\cdot} \left(\la b(s, X_s), \nabla f (X_s)\ra + \frac{1}{2} \textup{ trace } (a(s, X_s) \nabla^2 f(X_s)) \right)\dd s
	\end{split}
	\end{equation}
	is a local \((\F, P)\)-martingale. 
\end{definition}

From now on let \(P\) be a solution to the MP \((b, a, x_0)\). 
Our market is supposed to support \(\mathbb{N} \ni m \leq d\) risky assets. Each of them will be represented by a discounted asset price processes \(S^i = (S^i_t)_{t \in [0, T]}\).
Formally, for \(i = 1,..., m\), we set
\begin{align}
S^i_t \triangleq S_0^i\exp \left( \la e_i, X_t\ra - \la e_i, x_0\ra - \frac{1}{2} \int_0^{t} \la a(s, X_s) e_i, e_i\ra \dd s\right),\quad t \in [0, T],
\end{align}
for a deterministic initial value \(S_0^i > 0\).
This definition coincides with \eqref{eq: DPP}.

\begin{definition}
	We call a probability measure \(\q\) on \((\Omega, \mathcal{F})\) an \emph{equivalent (local) martingale measure (E(L)MM)}, if \(\q\sim\p\) and \(S \triangleq (S^1, ..., S^m)\) is a (local) \((\mathbf{F}, \q)\)-martingale.
\end{definition}
Thanks to the seminal work of \cite{DelbaenSchachermayer94}, the existence of an ELMM is equivalent to the (NFLVR) condition. Moreover, \cite{Cherny2007} showed that (NGA) is equivalent to the existence of an EMM.
\begin{definition}
	We call a strictly positive local \((\F, P)\)-martingale \(Z = (Z_t)_{t \in [0, T]}\) a \emph{strict (local) martingale density (S(L)MD)}, if the process \(Z S\) is a (local) \((\F, P)\)-martingale.
\end{definition}

%The \emph{strict} in this terminology refers to the positivity, by which we always mean strictly non-negative.
We stress that we use the convention that all local martingales have integrable initial values. 
\begin{remark}
	The existence of an E(L)MM implies the existence of a S(L)MD. Indeed, the density process of the E(L)MM is a S(L)MD, see Proposition III.3.4 in \cite{JS}. 
\end{remark}
The existence of a SLMD is equivalent to (NUPBR), see \cite{Choulli1996} and, in complete settings, the existence of a SMD is equivalent to (NRA), see \cite{fernholz2010}. For a profound discussion of weak notions of no arbitrage we refer to the article of \cite{doi:10.1142/S0219024915500053}.

\section{Conditions for the Absence of Arbitrage}\label{sec: ex}
The goal of this section is to derive deterministic conditions for the existence of E(L)MMs and SMDs. The section is divided into three parts. First, we give a recap of a classical drift condition which is equivalent to the existence of a SLMD, i.e. in particular necessary for the existence of an ELMM and a SMD, see Section \ref{sec: SLMD} below. Second, on the basis of the drift condition, 
we use non-explosion criteria to achieve our main goal and derive deterministic conditions for the existence of E(L)MMs and SMDs, see Section \ref{sec: det cond existence} below. In Section \ref{sec: comments existence} below we comment on the main results of this section and discuss examples. The proof of our main result is given in Section \ref{sec: pf main 1}.
\subsection{A Drift Condition}\label{sec: SLMD}

The following drift condition is well-known to be necessary for the existence of an ELMM, see, for instance, the article of \cite{protter}. 
To keep the article self-contained, we give a full proof.
As we will motivate below, it can be considered as a \emph{minimal} condition in the sense that it is equivalent to the existence of a SLMD.
\begin{lemma}\label{lem: MPR nece}
	If an ELMM exists, then there exists an \(\mathbb{R}^d\)-valued  \(\F\)-predictable process \(c = (c_t)_{t \in [0, T]}\) 
	such that \(P\)-a.s.
	\begin{align}\label{eq: c inte cond}
	\int_0^T \la a(s, X_s) c_s, c_s\ra \dd s < \infty,
	\end{align}
	and for all \(i = 1, ..., m\) and \(P\)-a.s. for all \(t \in [0,T]\)
	\begin{align}\label{eq: MPRE}
	\int_0^{t} \left \la e_i, b(s, X_s) + a(s, X_s) c_s\right \ra  \dd s= 0.
	\end{align}
\end{lemma}
\begin{proof}
	Let \(i \in \{1, ..., m\}\) be fixed and take an ELMM \(Q\).
	We denote by \(\mathcal{L}(\cdot)\) the stochastic logarithm.
	By Girsanov's theorem, see Theorem \ref{theo: Gir} in Appendix \ref{rem: SMP}, there exists an \(\mathbb{R}^d\)-valued \(\F\)-predictable process \(c = (c_t)_{t \in [0, T]}\) which satisfies \(Q\)-a.s. (and thus, by the equivalence of \(P\) and \(Q\), also \(P\)-a.s.) \eqref{eq: c inte cond} and \(Q\)-a.s.
	\begin{align}
	\mathcal{L} (S^i) = \int_0^{\cdot} \left\la e_i, b(s, X_s) + a(s, X_s) c_s\right\ra \dd s + \textup{local \((\F, Q)\)-martingale}.
	\end{align}
	Recall that the stochastic logarithm of a local martingale is also a local martingale. Hence, since the process \(S^i\) is a local \((\F, Q)\)-martingale, also \(\mathcal{L}(S^i)\) is a local \((\F, Q)\)-martingale. 
	Now, the claim of our lemma follows from the fact that continuous local martingales of finite variation are constant up to a null set, together with the equivalence of \(P\) and \(Q\).
\end{proof}

The equation \eqref{eq: MPRE} is called \emph{market price of risk equation (MPRE)} and a process \(c = (c_t)_{t \in [0, T]}\) as in Lemma \ref{lem: MPR nece} is called \emph{market price of risk (MPR)}.
\begin{remark}\label{rem: Ruf}
	As pointed out by \cite{MAFI:MAFI502}, if in a Markovian setting a MPR \(c = (c_t)_{t \in [0, T]}\) exists, it can be chosen with a Markovian structure, i.e. \(c_s (\omega) = \bar{c}(s, \omega(s))\) for a Borel function \(\bar{c} \colon [0,T] \times \mathbb{R}^d \to \mathbb{R}^d\). In particular, the Markovian version of the MPR is minimal in the sense that for any MPR \(c^* = (c^*_t)_{t \in [0, T]}\) we have \(P\)-a.s. \begin{align}\int_0^T \la a (s, X_s) \bar{c}(s, X_s), \bar{c}(s, X_s)\ra \dd s \leq \int_0^T\la a(s, X_s) c^*_s, c^*_s\ra \dd s.\end{align}
\end{remark}

The existence of a MPR is related to the so-called \emph{structure condition} as defined by \cite{Choulli1996}, which itself is 
equivalent to the existence of a SLMD, see Theorem 2.9 in \cite{Choulli1996}. Indeed, as we state in the following theorem, the existence of a MPR is also equivalent to the existence of a SLMD.
\newpage
\begin{theorem}\label{prop: NUPBR}
	The following is equivalent:
	\begin{enumerate}
		\item[\textup{(i)}]
		There exists a MPR.
		\item[\textup{(ii)}]
		There exists a SLMD.
	\end{enumerate}
\end{theorem}
\begin{proof}
	Assume that \(c = (c_t)_{t \in [0, T]}\) is a MPR and define 
	\begin{align}\label{eq: Z}
	Z_t \triangleq \exp \left( \int_0^t \la c_s, \dd X^c_s\ra - \frac{1}{2} \int_0^t \la a (s, X_s) c_s, c_s\ra \dd s\right), \quad t \in [0, T],
	\end{align}
	where \(X^c_t \triangleq X_t - x_0 - \int_0^t b(s, X_s)\dd s\), which is a continuous local \((\F, P)\)-martingale, see Proposition \ref{prop: eq MP} in Appendix \ref{rem: SMP}.
	This process is well-defined by \eqref{eq: c inte cond}. 
	Moreover, \(Z = (Z_t)_{t \in [0, T]}\) is a local \((\F, P)\)-martingale as it is the stochastic exponential of the local \((\F, P)\)-martingale \(\big(\int_0^t \la c_s, \dd X^c_s\ra\big)_{t \in [0, T]}\).
	Using the MPRE \eqref{eq: MPRE}, we obtain that
	\begin{align}\label{eq: SMD candi}
	S^i Z &= S^i_0 \exp \left(\int_0^\cdot \la c_s + e_i, \dd X_s^c\ra - \frac{1}{2} \int_0^\cdot \la a(s, X_s) (c_s + e_i), c_s + e_i\ra \dd s\right)
	\end{align}
	up to a \(P\)-null set. 
	In other words, \(S^i Z\) is the stochastic exponential of the local \((\F, P)\)-martingale \(\big(\int_0^t \la c_s + e_i, \dd X^c_s\ra\big)_{t \in [0, T]}\) and hence itself a local \((\F, P)\)-martingale.
	Thus, the implication (i) \(\Longrightarrow\) (ii) is proven.
	
	We prove the converse implication. Let \(Z = (Z_t)_{t \in [0, T]}\) be a SLMD. By the decomposition theorem of Kunita and Watanabe, see Theorem \ref{theo: deco KW} in Appendix \ref{rem: SMP}, we find an \(\mathbb{R}^d\)-valued \(\F\)-predictable process \(c = (c_t)_{t \in [0, T]}\) such that \eqref{eq: c inte cond} holds \(P\)-a.s., and a local \((\F, P)\)-martingale \(N = (N_t)_{t \in [0, T]}\) with \(\lle N, \la e_i,X^c\ra\rre = 0\) for all \(i = 1, ..., d\) such that 
	\begin{align}
	\mathcal{L} (Z)= \int_0^\cdot \la c_s, \dd X^c_s\ra + N.
	\end{align}
	Thus, by Equation II.8.20 in \cite{JS}, we obtain that
	\begin{equation}\begin{split}
	\mathcal{L} (S^i Z) &= \mathcal{L} (S^i) + \mathcal{L}(Z) + \lle \mathcal{L}(S^i), \mathcal{L}(Z)\rre
	\\&= \int_0^\cdot \la e_i, b(s, X_s)+ a(s, X_s) c_s\ra \dd s
	%\\&\hspace{3.25cm}
	+ \int_0^\cdot \la c_s + e_i, \dd X^c_s\ra + N.
	\end{split}
	\end{equation}
	Now, we can argue along the lines of the proof of Lemma \ref{lem: MPR nece} that the MPRE \eqref{eq: MPRE} has to hold up to \(P\)-null set. This completes the proof.
\end{proof}

\subsection{Deterministic Conditions for the Existence of SMDs and E(L)MMs}\label{sec: det cond existence}
In this section we assume the existence of a \emph{good version} of a MPR. We now explain what we mean by this. Define the \(\F\)-stopping time
\begin{align}\label{eq: tau n}
\tau_n \triangleq \inf(t \in [0,T] \colon \|X_t\| \geq n).
\end{align}
\begin{definition}
	We say that a MPR \(c = (c_t)_{t \in [0, T]}\) is \emph{good}, if for all \(n \in \mathbb{N}\)
	\begin{align}\label{eq: loc nov}
	E^P \left[ \exp \left(\frac{1}{2} \int_0^{T \wedge \tau_n} \la a(s, X_s) c_s, c_s \ra \dd s \right)\right] < \infty.
	\end{align}
\end{definition}
For \(x \in \mathbb{R}^d\) we set \(\|x\| \triangleq \sqrt{\la x, x\ra}\) and for \(A \in \mathbb{R}^d\otimes \mathbb{R}^d\) we set \(\|A\| \triangleq \sqrt{\textup{trace } (A A^*)}\), where \(A^*\) is the adjoint of \(A\).	
For a function \(c \colon [0, T] \times \mathbb{R}^d \to \mathbb{R}^{d-m}\) we write \((0,c)\) for the function \([0, T] \times \mathbb{R}^d \to \mathbb{R}^d\) with \(\la (0,c)(t, x), e_i\ra = 0\) if \(i \in \{1, ..., m\}\) and \(\la (0, c)(t, x), e_i\ra = \la c(t, x), e_{i - m}\ra\) for \(i \in \{m +1, ..., d\}\).
\begin{remark}\label{rem: M1}
	If \(a\) is invertible and \(b, a\) and \(a^{-1} b\) are locally bounded, then for any locally bounded Borel function~\(\mu \colon [0, T] \times \mathbb{R}^d \to \mathbb{R}^{d - m}\) the process 
	\begin{align}
	c_t (\omega) \triangleq - a^{-1} (t, \omega(t))b(t, \omega(t)) + (0, \mu (t, \omega(t))),\quad t \in [0, T], \omega \in \Omega,
	\end{align} 
	is a good MPR. 
\end{remark}
In this section we impose the
\begin{assumption}\label{Stand Ass}
	The function \((0, \mu)\colon [0, T] \times \mathbb{R}^d \to \mathbb{R}^d\) is such that \(\mu \colon [0, T] \times\mathbb{R}^{d} \to \mathbb{R}^{d - m}\) is a Borel function for which there is a good MPR \(c = (c_t)_{t \in [0, T]}\) such that \(P\)-a.s. 
	\begin{align}
	\int_0^T \|\mu(s, X_s)\| \dd s < \infty 
	\end{align}
	and for all \(t \in [0, T]\) and \(i = m + 1, ..., d\)
	\begin{align}\label{eq: set M}
	\int_0^{t} \la e_i, b (s, X_s) + a(s, X_s) c_s\ra \dd s = \int_0^{t} \la e_{i - m}, \mu(s, X_s) \ra \dd s.
	\end{align}
\end{assumption}
Clearly, this assumption holds in the setting of Remark \ref{rem: M1}. 
We now define several conditions which will imply the existence of SMDs and E(L)MMs and their uniqueness. 
The first two conditions can be seen as a (partial) multi-dimensional Feller test for explosion. It goes back to \cite{doi:10.1137/1105016} who stated it without proof.
Providing an intuition, the condition is based on a radial comparison with a one-dimensional diffusion.
\begin{condition}\label{cond: EL1}
	There exists an \(r > 0\) and a Borel function \(\zeta \colon [0, T] \to [0, \infty)\) such that \(\int_0^T \zeta (s)\dd s < \infty\) and for \(\dd t\)-a.a. \(t \in [0, T]\) and all \(x \in \{z \in \mathbb{R}^d \colon \|z\| \leq \sqrt{2 r}\}\)
	\begin{align}\label{loc bdd. gen}
	\|a(t, x)\| + \|\mu(t, x)\| \leq \zeta (t).
	\end{align}	
	Moreover, there are continuous functions \(A \colon [r, \infty) \to (0, \infty)\) and \(B \colon [r, \infty) \to (0, \infty)\) such that for \(\dd t\)-a.a. \(t \in [0, T]\) and all \(\rho \geq \sqrt{2 r}\) and \(x \in \{z \in \mathbb{R}^d\colon \|z\| = \rho\}\), 
	\begin{align}
	\zeta (t) A \left(\frac{\rho^2}{2}\right) &\geq \la a(t, x)x, x\ra,\label{cond bound 1}\\
	\zeta (t) \la a(t, x)x, x\ra B \left(\frac{\rho^2}{2}\right)&\geq \textup{ trace } a(t, x) + 2 \la x, (0, \mu)(t, x)\ra,\label{cond bound 2}\\
	&\hspace{-1.5cm}\int_r^\infty \frac{\int_r^z \frac{C(\sigma)}{A(\sigma)} \dd \sigma}{C(z)} \dd z = \infty,\label{cond: eq H}
	\end{align}
	where
	\begin{align}\label{C}
	C (z) \triangleq \exp \left(\int_r^z B(\sigma) \dd \sigma\right).
	\end{align}
\end{condition}

\begin{condition}\label{cond: E1}
	For all \(i = 1, ..., m\) the following holds:
	There exists an \(r_i > 0\) and a Borel function \(\zeta_i \colon [0, T] \to [0, \infty)\) such that \(\int_0^T \zeta_i (s)\dd s < \infty\) and \eqref{loc bdd. gen}, with \(\zeta\) replaced by \(\zeta_i\), holds for \(\dd t\)-a.a. \(t \in [0, T]\) and all \(x \in \{z \in \mathbb{R}^d \colon\|z\| \leq \sqrt{2 r_i}\}\).
	Moreover, there exist continuous functions \(A_i \colon [r_i, \infty) \to (0, \infty)\) and \(B_i \colon [r_i, \infty) \to (0, \infty)\) such that for \(\dd t\)-a.a. \(t \in [0, T]\) and all \(\rho \geq \sqrt{2 r_i}\) and \(x \in \{z \in \mathbb{R}^d \colon \|z\| = \rho\}\), 
	\begin{align}
	\zeta (t) A_i \left(\frac{\rho^2}{2}\right) &\geq \la a(t, x)x, x\ra,\\
	\zeta (t) \la a(t, x)x, x\ra  B_i \left(\frac{\rho^2}{2}\right)&\geq \textup{ trace } a(t,x) + 2 \la x, (0, \mu)(t, x) + a(t, x) e_i\ra,
	\end{align}
	and \eqref{cond: eq H} holds for \(A\) replaced by \(A_i\), \(B\) replaced by \(B_i\) and \(r\) replaced by \(r_i\).
\end{condition}

The following conditions are typically easy to verify. 
\begin{condition}\label{cond: EL3}
	There exists a Borel function \(\zeta \colon [0, T] \to [0, \infty)\) such that \(\int_0^T \zeta(s)\dd s <\infty\) and for \(\dd t\)-a.a. \(t \in [0, T]\) and all \(x \in \mathbb{R}^d\)
	\begin{align}
	\textup{trace } a(t, x) + 2 \la x, (0,\mu)(t, x)\ra \leq \zeta (t) (1 + \|x\|^2). 
	\end{align}
\end{condition}

\begin{condition}\label{cond: E3}
	There exists a Borel function \(\zeta \colon [0, T] \to [0, \infty)\) such that \(\int_0^T \zeta(s)\dd s <\infty\) and for all \(i = 1, ..., m\), \(\dd t\)-a.a. \(t \in [0, T]\) and all \(x \in \mathbb{R}^d\)
	\begin{align}
	\textup{trace } a(t, x)
	+ 2 \la x, (0, \mu)(t, x) + a (t, x) e_i\ra &\leq \zeta (t)(1 + \|x\|^2).
	\end{align}
\end{condition}
For the existence of an ELMM we also give an eigenvalue condition in the spirit of \cite{mckean1969stochastic}. We illustrate an application of this condition in Section \ref{section 4} below.
\begin{condition}\label{cond: EL MCKEAN}
	We have \(\mu \equiv 0\). Let \(\lambda_+ (a(t, x))\) be the largest eigenvalue of \(a(t, x)\) and set
	\begin{align}
	\gamma (z) \triangleq  \sup_{\|x\| \leq z}\sup_{t \in [0, T]} \lambda_+ (a(t, x)).
	\end{align}
	There exists a \(z_0 \geq \|x_0\|\) such that \(0 < \gamma (z) < \infty\) for all \(z \geq z_0\) and either
	\begin{align}
	\limsup_{n\to \infty} \frac{n^2}{\gamma(n)} = \infty
	\end{align}
	or there exists a Borel function \(\xi \colon [z_0, \infty) \to (0, \infty)\) such that \(\gamma(z) \leq \xi(z)\) for all \(z \geq z_0\), the map \(x \mapsto \frac{x}{\xi(x)}\) is Riemann integrable on \([z_0, n]\) for all \(\mathbb{N} \ni n > z_0\), and 
	\begin{align}\label{eq: cond MCKEAN}
	\limsup_{n \to \infty} \int_{z_0}^n \frac{x}{\xi (x)} \dd x = \infty.
	\end{align}
\end{condition}
We stress that \(\lambda_+\) is a Borel function, see \cite{10.2307/2039503}. 
Finally, we formulate two conditions which imply the uniqueness of the E(L)MM. %This is reflected by the acronym \emph{U}.
\begin{condition}\label{cond: U1}
	For all  \(x \in \mathbb{R}^d\) it holds that
	\begin{equation}\begin{split}\label{eq: uniq a cond}
	&\hspace{0.3cm}\inf_{s \in [0, T]} \inf_{\|\theta\| = 1} \la \theta, a(s, x) \theta\ra > 0,\\
	&\lim_{y \to x} \sup_{s \in [0, T]} \|a (s, y) - a(s, x)\| = 0.
	\end{split}
	\end{equation}
\end{condition}
\begin{condition}\label{cond: U2}
	The coefficient \(a\) has a root \(a^\frac{1}{2}\) such that for all \(n \in \mathbb{N}\) there exists a Borel function \(\zeta_n \colon [0, T] \to [0, \infty)\) such that \(\int_0^T \zeta_n(s)\dd s < \infty\) and for \(\dd t\)-a.a. \(t \in [0, T]\) and all \(x, y \in \{z \in \mathbb{R}^d\colon \|z\| \leq n\}\)
	\begin{align}
	\|a^\frac{1}{2}(t, x) - a^\frac{1}{2}(t, y)\|^2 \leq \zeta_n(t) \|x - y\|^2.
	\end{align}
\end{condition}

We are in the position to state the main theorem of this section. We prove the result in Section \ref{sec: pf main 1} below.
\begin{theorem}\label{theo: LG}
	\begin{enumerate}
		\item[\textup{(i)}]
		If one of the Conditions \ref{cond: EL1}, \ref{cond: EL3} and \ref{cond: EL MCKEAN} holds, then there exists a solution to the MP \(((0, \mu), a, x_0)\) which is also an ELMM. 
		Moreover, if in addition \(d = m\) and either Condition \ref{cond: U1} or Condition \ref{cond: U2} holds, then the ELMM is unique.
		\item[\textup{(ii)}]
		Suppose that there exists a Borel function \(\zeta\colon [0, T] \to [0, \infty)\) such that \(\int_0^T \zeta(s)\dd s < \infty\) and a locally bounded Borel function \(\hat{a} \colon [0, T] \times \mathbb{R}^d \to [0, \infty)\) such that for all \(\dd t\)-a.a. \(t \in [0, T]\) and all \(x \in \mathbb{R}^d\)
		\begin{align}\label{eq: a almost locally bounded}
		\max_{i = 1, ..., m} \langle a(t, x)e_i, e_i\rangle \leq \zeta (t) \hat{a}(t, x).
		\end{align}
		If one of the Conditions \ref{cond: E1} and \ref{cond: E3} is satisfied, then there exists a SMD.
		If in addition one of the Conditions \ref{cond: EL1}, \ref{cond: EL3} and \ref{cond: EL MCKEAN} holds, then there exists a solution to the MP \(((0, \mu), a, x_0)\) which is an EMM.
		Moreover, if in addition \(d = m\) and either Condition \ref{cond: U1} or Condition \ref{cond: U2} holds, then the EMM is unique.
	\end{enumerate}
\end{theorem}
Let us stress that no uniqueness argument is needed and that the result also applies in degenerated cases, see Example \ref{exp: Gir} below. 
\begin{remark}
	\begin{enumerate}
		\item[\textup{(i)}]
		Part (i) of Theorem \ref{theo: LG} stays true if the discounted asset process \(S\) is defined by \(S^i \triangleq \langle e_i, X\rangle\). 
		In this case, \(S^i\) may have non-positive paths. For our proof of Theorem \ref{theo: LG} (ii) the non-negativity of \(S^i\) is crucial,
		since we use a local change of measure with \(ZS^i\) and \(S^i\) as local densities. Here, \(Z\) is defined as in \eqref{eq: Z}. 
		\item[\textup{(ii)}] A version of Theorem \ref{theo: LG} also holds for the infinite time horizon. We discuss the adjustments of the conditions and the proof in Remark \ref{rem: infinite time horizon} below.
	\end{enumerate}
\end{remark}
An immediate consequence of Theorem \ref{theo: LG} is the following
\begin{corollary}\label{coro: EMM}
	If \(\mu\) and \(a\) are bounded uniformly on \([0, T] \times \mathbb{R}^d\), then there exists an EMM.
\end{corollary}
In the one-dimensional case Condition \ref{cond: U2} can be improved to local H\"older continuity. We state this special case formally in the following
\begin{proposition}\label{prop: uni 1d}
	Suppose that \(d = m = 1\). Moreover, assume that for all \(n \in \mathbb{N}\) there exists a strictly increasing function \(h_n \colon [0, \infty) \to [0, \infty)\) with \(h_n(0) = 0\) and a Borel function \(\zeta_n \colon [0, T] \to [0, \infty)\) such that for \(\dd t\)-a.a. \(t \in [0, T]\) and all \(\epsilon > 0\) and \(x, y \in [-n, n]\)
	\begin{align}
	\int_{0}^\epsilon \frac{1}{h_n(z)} \dd z &= \infty, \\
	\int_0^T \zeta_n (s)\dd s &< \infty,\\
	\|a^\frac{1}{2} (t, x) - a^\frac{1}{2} (t, y)\|^2 &\leq \zeta_n(t) h_n(\|x - y\|),	
	\end{align}
	then there exits at most one ELMM. In particular, there exists at most one EMM.
\end{proposition}
A proof is given in Appendix \ref{App: pf 1}. Before we comment on Theorem \ref{theo: LG}, let us state a version of Theorem \ref{theo: LG} for homogeneous coefficients.
\begin{corollary}
	Suppose that \(\mu\) and \(a\) are independent of time and locally bounded. Under each of the following two conditions a solution to the MP \(((0, \mu), a, x_0)\) exists which is also an ELMM.
	\begin{enumerate}
		\item[\textup{(i)}] There exist an \(r > 0\) and continuous functions \(A \colon [r, \infty) \to (0, \infty)\) and \(B \colon [r, \infty) \to (0, \infty)\) such that for all \(\rho \geq \sqrt{2 r}\) and \(x \in \{z \in \mathbb{R}^d \colon \|z\| = \rho\}\)
		\begin{align}
		A \left(\frac{\rho^2}{2}\right) &\geq \la a(x)x, x\ra,\\
		\la a(x)x, x\ra B \left(\frac{\rho^2}{2}\right)&\geq \textup{ trace } a(x) + 2 \la x, (0, \mu)(x)\ra,
		\end{align}
		and \eqref{cond: eq H} holds.
		\item[\textup{(ii)}]
		For all \(x \in \mathbb{R}^d\) we have 
		\begin{align}
		\textup{trace } a(x) + 2 \langle x, (0, \mu)(x) \rangle \leq \textup{const. } (1 + \|x\|^2).
		\end{align}
	\end{enumerate}
	Furthermore, under each of the following two conditions a SMD exists. 
	\begin{enumerate}
		\item[\textup{(iii)}] For all \(i = 1, ..., m\) there exist an \(r_i > 0\) and continuous functions \(A_i \colon [r_i, \infty) \to (0, \infty)\) and \(B_i \colon [r_i, \infty) \to (0, \infty)\) such that for all \(\rho \geq \sqrt{2 r_i}\) and \(x \in \{z \in \mathbb{R}^d \colon \|z\| = \rho\}\)
		\begin{align}
		A_i \left(\frac{\rho^2}{2}\right) &\geq \la a(x)x, x\ra,\\
		\la a(x)x, x\ra B_i \left(\frac{\rho^2}{2}\right)&\geq \textup{ trace } a(x) + 2 \la x, (0, \mu)(x) + a(x) e_i\ra,
		\end{align}
		and \eqref{cond: eq H} holds for \(A\) replaced by \(A_i, B\) replaced by \(B_i\) and \(r\) replaced by \(r_i\).
		\item[\textup{(iv)}]
		For all \(i = 1,..., m\) and \(x \in \mathbb{R}^d\) we have 
		\begin{align}
		\textup{trace } a(x) + 2 \langle x, (0, \mu)(x) + a(x) e_i \rangle \leq \textup{const. } (1 + \|x\|^2).
		\end{align}
	\end{enumerate}
	Moreover, if one of the conditions \textup{(i)} and \textup{(ii)} and one of the conditions \textup{(iii)} and \textup{(iv)} holds, then the solution to the MP \(((0, \mu), a, x_0)\) is an EMM.  
	If in addition \(d=m\) and \(a\) is either continuous and satisfies \(\inf_{\|\theta\| = 1}\langle \theta, a(x) \theta\rangle > 0\) for all \(x \in \mathbb{R}^d\), or has a locally Lipschitz continuous root, then the ELMM, respectively the EMM, is unique.
\end{corollary}

\subsection{Comments and Examples}\label{sec: comments existence}
Let us comment on consequences of Theorem \ref{theo: LG}.
We start with an example for a financial market in which not all solutions to the MP \((0, a, x_0)\) are ELMMs. Our result, however, yields that we can find a solution which is an ELMM. 
We stress that, as the example also illustrates, solutions can behave very differently.
\begin{example}\label{exp: Gir}
	Suppose that \(d = m = 1\) and that 
	\begin{align}\label{eq: Gir a}
	a(t, x) \triangleq \|x\|^\alpha \wedge 1, \quad \alpha \in (0, 1).
	\end{align}
	This diffusion coefficient corresponds to \emph{Girsanov's SDE} and it is well-known that the MP \((0, a, 0)\) has more than one solution.
	In fact, the Dirac measure on \(\{\omega \in \Omega \colon \omega(t) = 0 \text{ for all } t \in [0, T]\}\) is a solution and a non-Dirac solution is given by the law of \((B_{\tau_t})_{t \in [0, T]}\), where \((B_t)_{t \in [0, \infty)}\) is a one-dimensional Brownian motion and \((\tau_t)_{t \in [0, T]}\) is a time-change, see, for instance, \cite{RW2} pp. 175 for more details. Consequently, not all solutions to the MP \((0, a, 0)\) are equivalent. Since all ELMMs are equivalent, this implies that not all solutions can be ELMMs. However, if a good MPR exists for the market, by Corollary \ref{coro: EMM} we can always find an EMM in the set of solutions to the MP \((0, a, 0)\). 
\end{example}
If the MP \((b, a, x_0)\) has more than one solution, then the existence of an E(L)MM depends on the choice of \(P\).
To illustrate this, we give the following
\begin{example}\label{example: dependence P}
	Suppose that \(d = m = 1\) and that 
	\begin{align}
	b(t, x) \triangleq 3 x^{\frac{1}{3}},\quad a(t, x) \triangleq 9 \|x\|^{\frac{4}{3}}.
	\end{align}
	Then, it is clear that the Dirac measure on the set \(\{\omega \in \Omega \colon \omega(t) = 0 \text{ for all } t \in [0, T]\}\) is a solution to the MP \((b, a, 0)\). In fact, a non-Dirac solution is the law of \((W^3_t)_{t \in [0, T]}\), where \((W_t)_{t \in [0, T]}\) is a one-dimensional Brownian motion. This can easily be verified by It\^o's formula. Now, in the Dirac case the unique EMM is given by the real-world measure itself. If \(P\) is the law of \((W^3_t)_{t \in [0, T]}\), then there is no ELMM. To see this, note that, due to Lemma \ref{lem: uniqueness} below, any ELMM has to be a solution to the MP \((0, a, 0)\), which unique solution is the Dirac measure on the set \(\{\omega \in \Omega \colon \omega(t) = 0 \text{ for all } t \in [0, T]\}\). This measure, however, is not equivalent to the law of \((W^3_t)_{t \in [0, T]}\).
	Let us stress that this is no contradiction to Theorem \ref{theo: LG} since no MPR exists. This can be shown using the Engelbert-Schmidt zero-one law, see, for instance, Proposition 3.6.27 in \cite{KaraShre}. Therefore, not even a SLMD exists.
\end{example}

Theorem \ref{theo: LG} has interesting consequences for the incomplete case, which we discuss in the following

\begin{remark}
	Let us comment on the case \(m < d\). 
	Suppose that \(a\) is invertible, that \(b, a\) and \(a^{-1} b\) is locally bounded and that \begin{align}\textup{trace } a(t, x) \leq \textup{const. } (1 + \|x\|^2)\end{align} for \(\dd t\)-a.a. \(t \in [0, T]\) and all \(x \in \mathbb{R}^d\). In this case, for any locally bounded Borel function \(\mu \colon [0, T] \times \mathbb{R}^d \to \mathbb{R}^{d-m}\) our Assumption \ref{Stand Ass} holds, see Remark \ref{rem: M1}.
	If \(\mu\) is even bounded, then Condition \ref{cond: EL3} holds and Theorem \ref{theo: LG} yields that the MP \(((0, \mu), a, x_0)\) has a solution which is also an ELMM. Since there are infinitely many bounded functions \(\mu\), there are also infinitely many ELMMs.
\end{remark}

We also stress that, except of Assumption \ref{Stand Ass}, our conditions for the existence of an E(L)MM are independent of the coefficient \(b\). This is interesting when comparing our result to classical Novikov-type conditions which are typically imposed to ensure the existence of an E(L)MM. 
\begin{example}
	Suppose that \(m = d\), that \(b\) is locally bounded and that \(a \triangleq \textup{Id}\), where \(\textup{Id}\) denotes the identity matrix. 
	Then, \(c_s (\omega) \triangleq -b(s, \omega(s))\) is a good MPR and, by Corollary \ref{coro: EMM}, an EMM exists.
	In this setup, the classical Novikov condition for the existence of an ELMM is given by
	\begin{align}\label{NC}
	E^P \left[\exp\left( \frac{1}{2} \int_0^T \| b(s, X_s)\|^2 \dd s \right)\right] < \infty,
	\end{align}
	which is a condition particularly depending on \(b\). 
	For an example where \eqref{NC} is violated consider
	\begin{align}
	b(t, x_1, ..., x_d) \triangleq (C x_2, 0, ..., 0)^\text{tr},
	\end{align}
	for some large enough constant \(C > 0\).
	Then, it holds that
	\begin{align}
	E^P \left[ \exp \left(\frac{1}{2} \int_0^T \|b (s, X_s)\|^2 \dd s \right) \right] = E \left[ \exp \left(\frac{C^2}{2} \int_0^T W_s^2 \dd s \right)\right] = \infty,
	\end{align}
	where \((W_t)_{t \in [0, T]}\) is a one-dimensional Brownian motion.
	In this particular setting, Novikov's condition can be generalized, see the Corollaries 3.5.14 and 3.5.16 in \cite{KaraShre}.
	However, if the diffusion coefficient \(a\) is not constant, generalized Novikov-type conditions are not available, while Theorem \ref{theo: LG} may still apply.
\end{example}

\subsection{Proof of Theorem \ref{theo: LG}}\label{sec: pf main 1}
We first prove (i). 
%We set 
%\(
%X^c_t \triangleq X_t - x_0 -\int_0^{t} b(s, X_s) \dd s
%\) for \(t \in [0, T]\).
Let \(c = (c_t)_{t \in [0, T]}\) be a good MPR corresponding to \(\mu\), see Assumption \ref{Stand Ass}.
Then, we can define the local \((\F, P)\)-martingale \(Z\) as in \eqref{eq: Z}.
Since \(c\) is good, it follows from Novikov's condition that the stopped process \((Z_{t \wedge \tau_n})_{t \in [0, T]}\) is an \((\F, P)\)-martingale. 
We define a probability measure \(Q_n\) on \((\Omega, \mathcal{F})\) by the Radon-Nikodym derivative \(\dd Q_n = Z_{T \wedge \tau_n} \dd P\).
The following lemma is proven in Appendix \ref{App: pf 2}.
\begin{lemma}\label{lem: cm no explosion}
	If one of the Conditions \ref{cond: EL1}, \ref{cond: EL3} and \ref{cond: EL MCKEAN} holds, then \begin{align}\label{eq: Tul}\lim_{n \to \infty} Q_n(\tau_n > T) = 1.\end{align}
\end{lemma}
By the fact that \(\tau_n (\omega) \nearrow \infty\) as \(n \to \infty\) for all \(\omega \in \Omega\) and Lemma \ref{lem: cm no explosion}, we have 
\begin{align*}
E^P[Z_T] = \lim_{n \to \infty} E^P\left[Z_{T \wedge \tau_n} \1_{\{\tau_n> T\}}\right] = \lim_{n \to \infty} Q_n (\tau_n > T) = 1.
\end{align*}
Thus, \(Z = (Z_t)_{t \in [0, T]}\) is an \((\F, P)\)-martingale and \(\dd Q = Z_T \dd P\) defines a probability measure on \((\Omega, \mathcal{F})\) which is equivalent to \(P\).
By Girsanov's theorem, see Theorem \ref{theo: Gir} in Appendix \ref{rem: SMP}, \(Q\) solves the MP \(((0, \mu), a, x_0)\).
In particular, for all \(i = 1, ..., m\) the discounted asset price process \(S^i\) is a local \((\F, Q)\)-martingale and it follows that \(Q\) is an ELMM.

Let us now prove that in the case \(d = m\) either of the Conditions \ref{cond: U1} and \ref{cond: U2} implies that we have found the unique ELMM. 
The following lemma is also useful in the next section where we derive conditions for the non-existence of ELMMs.
\begin{lemma}\label{lem: uniqueness}
	If \(d = m\), then each ELMM solves the MP \((0, a, x_0)\).
\end{lemma}
\begin{proof}
	Suppose that \(Q\) is an ELMM. Then, \(S_0^{-1}S\), and thus also its stochastic logarithm \(X - x_0\), is a continuous local \((\F, Q)\)-martingale. 
	Since, by Girsanov's theorem, see Theorem \ref{theo: Gir} in Appendix \ref{rem: SMP}, the quadratic variation process of \(X\) is the same for \(P\) and \(Q\), it follows from Proposition \ref{prop: eq MP} in Appendix \ref{rem: SMP} that \(Q\) solves the MP \((0, a, x_0)\). 
\end{proof}
Due to this observation, the uniqueness of the ELMM follows immediately from the following lemma. A proof is given in Appendix \ref{App: pf 3}.
\begin{lemma}\label{lem: techn uni}
	If one of the Conditions \ref{cond: U1} and \ref{cond: U2} holds, then the set of solutions to the MP \((0, a, x_0)\) is a singleton. 
\end{lemma}

To prove the existence of a SMD or an EMM, it suffices to repeat the previous argument with \(ZS^i\) or \(S^i\) instead of \(Z\). 
The additional assumption \eqref{eq: a almost locally bounded} ensures that a local Novikov condition in the spirit of \eqref{eq: loc nov} holds.
We omit the details. Since, by (i), there is at most one ELMM under either of the Conditions \ref{cond: U1} and \ref{cond: U2}, there can obviously also be at most one EMM. The proof of Theorem \ref{theo: LG} is complete.
\qed
\begin{remark}\label{rem: infinite time horizon}
	A version of Theorem \ref{theo: LG} also holds on the infinite time horizon. Let us shortly point out which changes are necessary. 
	We implicitly assume that our setting is adjusted to the infinite time horizon. That means, for instance, that \(\Omega\) is the space of continuous functions \([0, \infty) \to \mathbb{R}^d\), equipped with the local uniform topology, and \(\mathcal{F}\) is the corresponding Borel \(\sigma\)-field. In particular, the Assumption \ref{Stand Ass} has to hold for all \(T >0\). Then, we can define a probability measure \(Q_n\) by the Radon-Nikodym derivative \(\dd Q_n = Z_{n \wedge \tau_n} \dd P\).
	If one of the Conditions \ref{cond: EL1}, \ref{cond: EL3} and \ref{cond: EL MCKEAN} holds for all \(T > 0\), then \eqref{eq: Tul} holds for all \(T > 0\) and an application of Tulcea's extension theorem, see Theorem 1.3.5 in \cite{SV}, yields the existence of a probability measure on \((\Omega, \mathcal{F})\) which coincides with \(Q_n\) on \(\mathcal{F}_{n \wedge \tau_n}\). This extension is locally equivalent of \(P\), see Lemma III.3.3 in \cite{JS}, and such that \(\langle X, e_i\rangle\) is a local martingale. Similarly, 
	if (in addition) the Conditions \ref{cond: E1} and \ref{cond: E3} hold for all \(T > 0\), then \(Z S^i\) is a \(P\)-martingale (\(S^i\) is a martingale w.r.t. the extension). 
\end{remark}

\section{Deterministic Conditions for the Non-Existence of SMDs and E(L)MMs}\label{sec: nonex}
In this section we suppose that \(d = m\) and, for simplicity, \(\|x_0\| = 1\). 
The following conditions can be viewed as complements to the Conditions \ref{cond: EL1} and \ref{cond: E1}. They are extensions of criteria for explosion by \cite{doi:10.1137/1105016} as given in Theorem 10.2.4 in \cite{SV} and imply that the MP \((0, a, x)\), respectively the MP \((a e_i, a, x)\), has no solution.  

We stress that Theorem 10.2.4 in \cite{SV} only implies that the MPs \((0, a, x)\) and \((a e_i, a, x_0)\) have no solutions for large enough time horizons. We establish this result for an arbitrary finite time horizon.
\begin{condition}\label{cond: NL1}
	There exist continuous functions \(A \colon (0, \infty) \to (0, \infty)\) and \(B \colon (0, \infty) \to (0, \infty)\) such that for all \(t \in [0, T], \rho >0\) and \(x \in \{z \in \mathbb{R}^d\colon \|z\| = \rho\}\) 
	\begin{align}
	A \left(\frac{\rho^2}{2}\right) &\leq \la a(t, x)x, x\ra,\label{A1}\\
	\la a(t, x)x, x\ra B \left(\frac{\rho^2}{2}\right)&\leq \textup{ trace } a(t, x),\label{B1}\\
	&\hspace{-1.5cm}\int_{\frac{1}{2}}^\infty \frac{\int_{\frac{1}{2}}^z \frac{C(\sigma)}{A(\sigma)} \dd \sigma}{C(z)} \dd z < \infty,\label{N MM 1}\\
	&\hspace{-1.5cm}\int_0^{\frac{1}{2}} \frac{\int_z^{\frac{1}{2}} \frac{C(\sigma)}{A(\sigma)} \dd \sigma}{C(z)} \dd z = \infty\label{N MM 2},
	\end{align}
	where
	\begin{align}
	C (z) \triangleq \exp \left(\int_{\frac{1}{2}}^z B(\sigma) \dd \sigma\right).
	\end{align}
	Moreover, for all \(n \in \mathbb{N}\) there exist a strictly increasing function \(\rho_n \colon [0, \infty) \to [0, \infty)\) with \(\rho_n(0) = 0\) and a strictly increasing, concave and continuous function \(\kappa_n \colon [0, \infty) \to [0, \infty)\) with \(\kappa_n(0) = 0\) such that for all \(x, y \in [\frac{1}{n}, n]\) it holds that 
	\begin{align}\label{gamma n}
	\|A^{\frac{1}{2}} (x) - A^{\frac{1}{2}} (y)\|^2 &\leq \rho_n (\|x - y\|),\\
	\label{kappa n}
	\| A(x) B(x) - A(y) B(y)\| &\leq \kappa_n(\|x - y\|),
	\end{align}
	and for all \(\epsilon > 0\)
	\begin{align} \label{rho cond}
	\int_0^\epsilon \frac{1}{\rho_n(z)}\dd z &= \infty,\\
	\int_0^\epsilon \frac{1}{\kappa_n(z)} \dd z &= \infty.\label{kappa cond}
	\end{align}
\end{condition}
\begin{remark}
	The second part of Condition \ref{cond: NL1} is in the spirit of the conditions for pathwise uniqueness as given by \cite{yamada1971}.
	For instance, one may choose \(\rho_n (x)=\kappa_n (x) = x\) if \(A^{\frac{1}{2}}\) is locally H\"older continuous with exponent \(\frac{1}{2}\) and \(AB\) is locally Lipschitz continuous.
\end{remark}
We also formulate a condition for non-existence of a SMD. Recall that if no SMD exists, then also no EMM exists.
\begin{condition}\label{cond: N1}
	For at least one \(i \in \{ 1, ..., d\}\) the following holds:
	There exist continuous functions \(A \colon (0, \infty) \to (0, \infty)\) and \(B \colon (0, \infty) \to (0, \infty)\) such that for all \(t \in [0, T], \rho >0\) and \(x \in \{z \in \mathbb{R}^d \colon\|z\| = \rho\}\), 
	\begin{align}
	A \left(\frac{\rho^2}{2}\right) &\leq \la a(t, x)x, x\ra,\\
	\la a(t, x)x, x\ra B \left(\frac{\rho^2}{2}\right)&\leq \textup{ trace } a(t, x) + 2 \la x, a(t, x) e_i\ra,\label{eq: 4.11}
	\end{align}
	and \eqref{N MM 1} and \eqref{N MM 2} hold. 
	Moreover, for all \(n \in \mathbb{N}\) there exist a strictly increasing function \(\rho_n \colon [0, \infty) \to [0, \infty)\) with \(\rho_n(0) = 0\) and a strictly increasing, concave and continuous function \(\kappa_n \colon [0, \infty) \to [0, \infty)\) with \(\kappa_n(0) = 0\) such that \eqref{gamma n} and \eqref{kappa n} hold for all \(x, y \in [\frac{1}{n}, n]\)
	and \eqref{rho cond} and \eqref{kappa cond} hold for all \(\epsilon > 0\).
\end{condition}
The proof of the following theorem is based on a comparison argument in the spirit of the proof of Khasminskii's criterion for explosion as given by \cite{ikeda1977} 
together with the fact that one-dimensional diffusions explode arbitrarily fast.
\begin{theorem}\label{theo: det NE 1}
	\begin{enumerate}
		\item[\textup{(i)}] Suppose that Condition \ref{cond: NL1} holds, then no ELMM exists.
		\item[\textup{(ii)}] Suppose that Condition \ref{cond: N1} holds, then no SMD exists. In particular, if Condition \ref{cond: N1} holds, then no EMM exists.
	\end{enumerate}
\end{theorem}
\begin{remark}
	The Theorems \ref{prop: NUPBR} and \ref{theo: det NE 1} can be used to construct financial markets in which a SLMD exists, while no ELMM exists. We give an example in Section \ref{section 4} below. Similarly, the Theorems \ref{prop: NUPBR} and \ref{theo: det NE 1} can be used to construct a financial market in which a SMD exists, while no ELMM exists.
	
	Theorem \ref{theo: det NE 1} in combination with Theorem \ref{theo: LG} can be used to construct financial markets in which an ELMM exists, but no EMM.
	In this case the market includes a financial bubble as defined by \cite{Cox2005}, see also Section \ref{section 4} below.
	
	Furthermore, if a MPR exists and Condition \ref{cond: N1} holds, then there exists a SLMD which is no SMD.
\end{remark}
\noindent
\emph{Proof of Theorem \ref{theo: det NE 1}:}
By Lemma \ref{lem: uniqueness}, part (i) follows if Condition \ref{cond: NL1} implies that the MP \((0, a, x_0)\) has no solution. 
For the second part, note the following 
\begin{lemma}\label{lem: no EMM}
	If a SMD exists, then the MP \((a e_i, a, x_0)\) has a solution for all \(i = 1, ..., d\).
\end{lemma}
\begin{proof}
	Suppose there exists a SMD \(Z = (Z_t)_{t \in [0, T]}\).
	As shown in the proof of Theorem \ref{prop: NUPBR}, there exists a MPR \(c = (c_t)_{t \in [0, T]}\) and a local \((\F, P)\)-martingale \(N = (N_t)_{t \in [0, T]}\) with \(\lle N, \la e_i,X^c\ra\rre = 0\) for all \(i = 1, ..., d\) such that 
	\begin{equation}\begin{split}
	\mathcal{L}\left(Z S^i\right) = \int_0^\cdot \la c_s + e_i, \dd X^c_s\ra + N.
	\end{split}
	\end{equation}
	Denote by \(\mathcal{E}(\cdot)\) the stochastic exponential. We have \begin{align}\mathcal{E}\left(\mathcal{L}\left(Z S^i\right)\right) = \frac{ZS^i}{Z_0S^i_0},\end{align}
	see Corollary II.8.7 in \cite{JS}.
	Thus, it holds that 
	\begin{align}\label{eq: G appli}
	\lle \la X, e_i\ra, Z S^i\rre = \int_0^\cdot Z_{s-} S^i_s \langle e_i, a(s, X_s)(c_s + e_i) \rangle \dd s.
	\end{align}
	By definition, \(ZS^i\) is an \((\F, P)\)-martingale.
	Thus, we can define a probability measure \(Q^i\) on \((\Omega, \mathcal{F})\) by the Radon-Nikodym derivative \(\dd Q^i = Z_TS^i_T/Z_0S_0^i \dd P\). 
	Now, Girsanov's theorem, see Theorem \ref{theo: Gir} in Appendix \ref{rem: SMP}, Proposition \ref{prop: eq MP} in Appendix \ref{rem: SMP}, \eqref{eq: MPRE} and \eqref{eq: G appli} yield that \(Q^i\) solves the MP \((a e_i, a, x_0)\). 
\end{proof}

Thus, since the existence of an EMM implies the existence of a SMD, part (ii) follows if Condition \ref{cond: N1} implies that the MP \((a e_i, a, x_0)\) has no solution for some \(i \in \{1, ..., d\}\). The argument is identical for both cases, so that we prove them simultaneously. 
Let 
\begin{align}
\hat{b} \triangleq \begin{cases}
a e_i&\text{ if Condition \ref{cond: NL1} holds},\\
0,&\text{ if Condition \ref{cond: N1} holds}.
\end{cases}
\end{align}
We use a contradiction argument.
Suppose that the MP \((\hat{b}, a, x_0)\) has a solution. 
Since the existence of weak solutions to SDEs is equivalent to the existence of solutions to MPs, see Theorem \ref{theo: equivalence MP SDE up to expl} in Appendix \ref{Appendix}, there exists a solution process \(Y = (Y_t)_{t \in [0, T]}\) to the SDE
\begin{align}
\dd Y_t = \hat{b}(t, Y_t)\dd t + a^{\frac{1}{2}}(t, Y_t) \dd W_t, \quad Y_0 = x_0,
\end{align}
where \(W = (W_t)_{t \in [0, T]}\) is a \(d\)-dimensional Brownian motion.
Define  \(p(x) \triangleq \frac{1}{2} \|x\|^2\) and
\begin{align}
\phi_t \triangleq \inf\left(s \in [0, T] \colon \int_0^s \frac{\la a(r, Y_r) Y_r, Y_r\ra}{A(p(Y_r))} \dd r \geq t\right).
\end{align}
Note that, by \eqref{A1}, \(\phi_t \leq t\).
Set 
\begin{align}
u(t, x) \triangleq \la x, \hat{b} (t, x)\ra + \frac{1}{2} \text{ trace } a(t, x).
\end{align}
Using It\^o's formula, we obtain 
\begin{align}
\dd p(Y_t) =  \la Y_t, 	a^{\frac{1}{2}} (t, Y_t)\dd W_t\ra + u(t, Y_t)\dd t,\quad p(Y_0) = \frac{1}{2}.
\end{align}
The map
\begin{align} t \mapsto \int_0^t \frac{\la a(r, Y_r) Y_r, Y_r\ra}{A(p(Y_r))}\dd r\end{align} is continuous and strictly increasing. Thus, also \(t \mapsto \phi_t\) is continuous and strictly increasing and, by classical rules for time-changed (stochastic) integrals, see Appendix \ref{Appendix: Time Change}, we obtain
\begin{equation}\begin{split}
\dd p(Y_{\phi_t}) &= \la Y_{\phi_t}, a^{\frac{1}{2}} (\phi_t, Y_{\phi_t}) \dd W_{\phi_t}\ra + u(\phi_t, Y_{\phi_t}) \dd \phi_t
\\&= \la Y_{\phi_t}, a^{\frac{1}{2}} (\phi_t, Y_{\phi_t}) \dd W_{\phi_t}\ra + \frac{1}{2} \frac{ A(p(Y_{\phi_t}))u(\phi_t, Y_{\phi_t})}{ \la a(\phi_t, Y_{\phi_t}) Y_{\phi_t}, Y_{\phi_t}\ra} \dd t.
\end{split}
\end{equation}
We set 
\begin{align}
B_t \triangleq \int_0^t \la A^{- \frac{1}{2}} (p(Y_{\phi_t})) Y_{\phi_t}, a^{\frac{1}{2}} (\phi_t, Y_{\phi_t}) \dd W_{\phi_t}\ra,\quad t \in [0, T].
\end{align}
Due to Proposition \ref{prop: TC Mg} in Appendix \ref{Appendix: Time Change} the process \(B = (B_t)_{t \in [0, T]}\) is a local martingale (w.r.t. a time-changed filtration) with
\begin{equation}\begin{split}
\lle B \rre_t &= \int_0^t \frac{\la a (\phi_s, Y_{\phi_s}) Y_{\phi_s}, Y_{\phi_s}\ra}{A (p (Y_{\phi_s}))} \dd\hspace{0.06cm} \lle W_{\phi_\cdot}\rre_s
\\&= \int_0^t \frac{\la a (\phi_s, Y_{\phi_s}) Y_{\phi_s}, Y_{\phi_s}\ra}{A (p (Y_{\phi_s}))} \dd\phi_s = t.
\end{split}
\end{equation}
Thus, by L\'evy's characterization of Brownian motion, see, for instance, Theorem IV.3.6 in \cite{RY}, the process \(B\) is a Brownian motion. 
We have found the dynamics
\begin{align}
\dd p(Y_{\phi_t}) = A^{\frac{1}{2}} (p(Y_{\phi_t})) \dd B_t + \frac{A(p(Y_{\phi_t})) u(\phi_t, Y_{\phi_t})}{\la a(\phi_t, Y_{\phi_t}) Y_{\phi_t}, Y_{\phi_t}\ra} \dd t,\quad p(Y_{\phi_0}) = \frac{1}{2}.
\end{align}

Let \(I_\Delta \triangleq (0, \infty) \cup \{\Delta\}\) be the one-point compactification of \(I \triangleq (0, \infty)\). 
As in Appendix \ref{Appendix}, we denote by \(\mathbb{W}\) the space of all continuous functions \(\omega\colon [0, \infty) \to I_\Delta\) such that \(\omega(t) = \Delta\) for all \(t \geq \gamma_I(\omega) \triangleq \inf(t \in [0, \infty)\colon \omega(t) \not \in I)\).
Furthermore, we set
\begin{align}
\gamma_n (\omega) \triangleq \inf\left(t \in [0, \infty) \colon \omega(t) \not \in \left(\frac{1}{n}, n\right) \right),\quad \omega \in \mathbb{W}.
\end{align}
The following lemma is proven in Appendix \ref{App: pf 4}.

\begin{lemma}\label{lem: Z compactification}
	There exists a process \(Z = (Z_t)_{t \in [0, \infty)}\) with paths in \(\mathbb{W}\) such that for all \(t \in [0, \infty)\) and \(n \in \mathbb{N}\)
	\begin{equation}\label{SDE: Z}\begin{split}
	Z_{t \wedge \gamma_n(Z)} &= \frac{1}{2} + \int_0^{t \wedge \gamma_n(Z)} A^{\frac{1}{2}}(Z_s) \dd B_s + \frac{1}{2} \int_0^{t \wedge \gamma_n(Z)} A(Z_s) B(Z_s)\dd s.
	\end{split}
	\end{equation}
\end{lemma}

Next, we compare \((Z_t)_{t \in [0, T]}\) and \(U = (U_t)_{t \in [0, T]} \triangleq (p(Y_{\phi_t}))_{t \in [0, T]}\).
We set 
\begin{align}
\xi_n \triangleq \inf\left(t \in [0, T] \colon U_t \not \in \left(\frac{1}{n}, n\right)\right).
\end{align}
The following lemma is proven in Appendix \ref{App: pf 5}.
\begin{lemma}\label{lem: comparison}
	Almost surely for all \(n \in \mathbb{N}\) and \(t \in [0, T]\)
	\begin{align}\label{comparison}
	Z_{t \wedge \gamma_n(Z) \wedge \xi_n} \leq U_{t \wedge \gamma_n(Z) \wedge \xi_n}.
	\end{align}
\end{lemma}
The previous lemma states that a.s. the paths of \((Z_t)_{t \in [0, T]}\) are below the paths of \(U\) till either \((Z_t)_{t \in [0, T]}\) or \(U\) leaves \((0, \infty)\).
However, by the Feller test for explosion, see Theorem \ref{theo: Feller test} in Appendix \ref{Appendix}, the condition \eqref{N MM 2} yields that \(Z\) a.s. does not explode to \(0\). Hence, by the pathwise ordering of \((Z_t)_{t \in [0, T]}\) and \(U\), also \(U\) a.s. does not explode to \(0\). By Feller's test and Proposition \ref{prop: fast explosion} in Appendix \ref{Appendix}, \eqref{N MM 1} implies that \(Z\) explodes to \(+ \infty\) arbitrarily fast, i.e. before time \(T\) with positive probability.
Therefore, our pathwise ordering implies that \(U\) explodes to \(+ \infty\) with positive probability.
We deduce from \(\phi_t \leq t\) that \(\|Y\|\) explodes to \(+ \infty\) with positive probability. This, however, is a contradiction to the fact that \(Y\) has paths in \(\Omega\). Hence, we conclude that the MP \((0, a, x_0)\) has no solution. 
This completes the proof.
\qed
\section{The Influence of the Market Dimension}\label{section 4}
The number of assets in a market plays a crucial role for the existence of an ELMM.
In this section we give an example for a financial market in which the number of risky assets coincides with the number of the sources of risk and the existence of an ELMM depends on the number of assets.

Let \(f \colon \mathbb{R}^d\to (0, \infty)\) be a locally bounded Borel function such that for all \(R > 0\)
\begin{align}
\inf_{\|x\| \leq R} f(x) > 0.
\end{align}
We suppose that \(m = d\), that there is a Borel function \(\zeta \colon [0, T] \to [0, \infty)\) and a locally bounded Borel function \(\hat{b}\colon[0, T] \times \mathbb{R}^d \to [0, \infty)\) such that \(\int_0^T \zeta^2(s)\dd s <\infty\), \(\|b(t, x)\| \leq \zeta (t) \hat{b}(t, x)\),  \(\|x_0\| = 1\) and
\begin{align}
\la a(t, x)e_i, e_j\ra \triangleq f(x) \1_{\{i = j\}}.
\end{align}
In this case, \(c_s(\omega) \triangleq - a^{-1} (s, \omega(s)) b (s, \omega(s))\) is a good MPR. In particular, by Theorem \ref{prop: NUPBR}, there always exists a SLMD.
However, as we will see next, if the dimension is high enough there are cases in which no ELMM exists.
We transfer the Conditions \ref{cond: EL MCKEAN} and \ref{cond: NL1} to this setting.
\begin{condition}\label{cond example existence}
	There exists a continuous function \(\xi \colon [1, \infty) \to (0, \infty)\) such that \(\xi(z) \geq \sup_{\|x\| \leq z} f(x)\) for all \(z \in [1, \infty)\) and 
	\begin{align}
	\int_1^\infty \frac{\rho}{\xi(\rho)} \dd \rho = \infty.
	\end{align}
\end{condition}
\begin{condition}\label{cond example}
	There exists a locally Lipschitz continuous function \(\alpha\colon (0, \infty) \to (0, \infty)\) such that \(\alpha(\rho) \leq f(x)\) for \(\rho> 0\) and \(x \in \{z \in \mathbb{R}^d \colon \|z\| = \rho\}\), and
	\begin{align}\label{conv 2}
	\int_1^\infty \frac{\rho}{\alpha(\rho)}\dd \rho < \infty.
	\end{align}
\end{condition}
\begin{proposition}\label{dim coro}
	If \(d \leq 2\), then there exists a unique ELMM.
	Moreover, if \(d \geq 3\) and Condition \ref{cond example existence} holds, then there is an ELMM, which is unique if \(f\) is continuous. But if \(d \geq 3\) and Condition \ref{cond example} holds, then there is no ELMM.
\end{proposition} 
\begin{proof}
	Let \(d \leq2\). 
	Set \(Q_n\) as in the proof of Theorem \ref{theo: LG} and note, by Lemma \ref{lem: Qn} in Appendix \ref{App: pf 2}, that \(Q_n\) solves the MP \((0, a, x_0; \tau_n)\), see Definition \ref{def: stopped MP} in Appendix \ref{App: pf 2} for this notation. 
	If we show that \(\lim_{n \to \infty} Q_n(\tau_n > T) = 1\), then the existence of an ELMM follows as in the proof of Theorem \ref{theo: LG}.
	By Exercise 10.3.3 in \cite{SV}, for all initial valued \(x \in \mathbb{R}^d\) the MP \((0, a, x)\) up to explosion has a unique non-explosive solution \(Q_x\).
	Thus, by Proposition \ref{prop: loc uni} in Appendix \ref{Appendix} we have \(Q_{x_0} = Q_n\) on \(\mathcal{F}_{\tau_n}\) and, since \(Q_{x_0}\) is non-explosive, it holds that
	\begin{align}
	\lim_{n \to \infty} Q_n (\tau_n > T) = \lim_{n \to \infty} Q_{x_0}(\tau_n > T) = 1.
	\end{align}
	We conclude that an ELMM exists. Since, again due to Proposition \ref{prop: loc uni} in Appendix \ref{Appendix}, the MP \((0, a, x_0)\) has only one solution, Lemma \ref{lem: uniqueness} yields that the ELMM is unique.
	
	Clearly, Condition \ref{cond example existence} implies Condition \ref{cond: EL MCKEAN}. Thus, if \(d \geq 3\) and Condition \ref{cond example existence} holds, an ELMM exists by Theorem \ref{theo: LG}. Moreover, if \(f\) is continuous, Condition \ref{cond: U1} holds and Theorem \ref{theo: LG} implies furthermore that the ELMM is unique.
	
	Suppose now that \(d \geq 3\) and that Condition \ref{cond example} holds. 
	Set \(A(x) = 2x \alpha(\sqrt{2 x})\) and \(B(x) = \frac{d}{2x}\) for \(x \in (0, \infty)\). It is routine to check that \(A\) and \(B\) satisfy \eqref{A1} and \eqref{B1}.
	Moreover, since compositions and products of locally Lipschitz continuous functions are locally Lipschitz continuous, \(A^\frac{1}{2}\) and \(AB\) are locally Lipschitz continuous. 
	Let \(C\) be defined as in \eqref{C}.
	By Fubini's theorem, using that \(d \geq 3\), we obtain
	\begin{equation}\begin{split}
	\int_{\frac{1}{2}}^{\infty} \frac{\int_{\frac{1}{2}}^z \frac{C(\sigma)}{A(\sigma)}\dd \sigma}{C(z)} \dd z &= \int_{\frac{1}{2}}^{\infty} \int_{\sigma}^{\infty} \frac{1}{C(z)} \dd z \frac{C(\sigma)}{A(\sigma)} \dd \sigma 
	\\&= \textup{ const. } \int_{\frac{1}{2}}^{\infty} \frac{\sigma}{A(\sigma)}\dd \sigma 
	\\&= \textup{ const. } \int_{\frac{1}{2}}^\infty \frac{1}{\alpha(\sqrt{2 \sigma})} \dd \sigma
	\\&= \textup{ const. } \int_{1}^\infty \frac{\sigma}{\alpha(\sigma)} \dd \sigma.
	\end{split}
	\end{equation}
	Hence, \eqref{conv 2} is equivalent to \eqref{N MM 1}. Moreover, using again \(d \geq 3\), we obtain
	\begin{align} %\label{dim influence indication}
	\int_0^{1/2} \exp\left( - \int_{1/2}^z B(u)\dd u\right) \dd z = \textup{const. }  \int_0^{1/2} \frac{1}{z^{d/2}}\dd z = + \infty.
	\end{align}
	Thus, Problem 5.5.27 in \cite{KaraShre} implies that \eqref{N MM 2} is satisfied. 
	Putting these pieces together, Condition \ref{cond example} implies Condition \ref{cond: NL1}.
	Therefore, due to Theorem \ref{theo: det NE 1}, no ELMM exists.
\end{proof}
Without further assumptions on \(f\) we cannot conclude uniqueness of the ELMM in the case \(d \geq 3\), see \cite{zbMATH01140123} for more details.

The Conditions \ref{cond example existence} and \ref{cond example} are very close. We think that in the present market they provide a fairly complete picture of the existence respectively non-existence of ELMMs for high dimensions.

Proposition \ref{dim coro} yields the existence of financial markets including arbitrage opportunities where the number of assets coincides with the number of sources of risk.
We give an explicit example.
\begin{example}
	Let \(d \geq 3\) and
	suppose that \(f(x) \triangleq (\|x\| \vee 1)^{2 + \epsilon}\) for some \(\epsilon > 0\) and set \(b(x) = - \frac{\beta}{2} \|x\|^\epsilon x\) with \(\beta \geq d - 2\). In this case, the MP \((b, a, x_0)\) has a unique solution, see Exercise 10.3.4 in \cite{SV}. Now, we can take \(\alpha (\rho) \triangleq \rho^{2 + \epsilon}\) which is a locally Lipschitz continuous function. It is easy to check that \eqref{conv 2} holds. Therefore, by Theorem \ref{prop: NUPBR} and Proposition \ref{dim coro}, there exists a SLMD, but no ELMM.
\end{example}
If \(d \leq 2\) there is always an ELMM, but there are case in which no EMM exists, i.e. the market includes a financial bubble. 
Let us illustrate this for the one-dimensional case. 
\begin{proposition}\label{proposition: 1D}
	Assume that \(d = m = 1\). Then, the following are equivalent:
	\begin{enumerate}
		\item[\textup{(i)}] A SMD exists.
		\item[\textup{(ii)}] An EMM exists.
		\item[\textup{(iii)}] We have
		\begin{align}\label{eq: Feller cond 1D}
		\int_1^\infty \frac{1}{f(y)}\dd y = \infty. 
		\end{align}
	\end{enumerate}
	Moreover, if an EMM exists, it is unique.
\end{proposition}
\begin{proof}
	We prove (i) \(\Longrightarrow\) (iii) \(\Longrightarrow\) (ii). This suffices since (ii) \(\Longrightarrow\) (i) is clear.
	
	First, assume that (i) holds. Then, by Lemma \ref{lem: no EMM}, the MP \((f, f, x_0)\) has a solution. By Theorem \ref{theo: Feller test} and Remark \ref{rem: explosion} in Appendix \ref{Appendix} this implies (iii). To see this, note that Fubini's theorem implies that 
	\begin{align}
	\int_0^\infty \exp(- 2 x) \int_0^x \frac{\exp(2 y)}{f(y)}\dd y\dd x = \frac{1}{2} \int_0^\infty \frac{1}{f(x)}\dd x,
	\end{align}
	and that 
	\begin{align}
	\int_{- \infty}^0 \exp(- 2 x) \int_x^0 \frac{\exp(2 y)}{f(y)}\dd y\dd x = \infty.
	\end{align}
	
	Let us now suppose that (iii) holds. Due to Proposition \ref{dim coro}, there exists a unique ELMM \(Q\) and thus there is at most one EMM.
	By Novikov's condition, \((S_{t \wedge \tau_n})_{t \in [0, T]}\) is an \((\F, Q)\)-martingale and we can define a probability measure \(Q_n\) on \((\Omega, \mathcal{F})\) by the Radon-Nikodym derivative \(\dd Q_n \triangleq S_{T \wedge \tau_n}/S_0 \dd Q\). It follows as in the proof of Lemma \ref{lem: Qn} in Appendix \ref{App: pf 2} that \(Q_n\) solves the MP \((f, f, x_0; \tau_n)\), see Definition \ref{def: stopped MP} in Appendix \ref{App: pf 2} for this notation. Since (iii) holds, Theorem \ref{theo: Feller test} in Appendix \ref{Appendix} yields that for all initial values \(x \in \mathbb{R}\) the MP \((f, f, x)\) (up to explosion) has a unique non-explosive solution \(Q_x\). 
	We deduce from Proposition \ref{prop: loc uni} in Appendix \ref{Appendix} that \(Q_n = Q_{x_0}\) on \(\mathcal{F}_{\tau_n}\). Thus, we have 
	\begin{equation}\begin{split}
	E^Q\left[\frac{S_T}{S_0}\right] &= \lim_{n \to \infty} E^Q\left[\frac{S_{T \wedge \tau_n}}{S_0} \1_{\{ \tau_n > T\}}\right] \\&= \lim_{n \to \infty} Q_n ( \tau_n >T) \\&= \lim_{n \to \infty} Q_{x_0}(\tau_n  >T ) \\&= 1,
	\end{split}
	\end{equation}
	which implies that \(S\) is an \((\F, Q)\)-martingale, i.e. that \(Q\) is an EMM.
\end{proof}
\begin{example}
	Let \(d =1\) and \(\delta > 0\). Then, for \(f (x) \triangleq (\|x\| \vee 1)^\delta\) a SMD exists \(\Longleftrightarrow\) an EMM exists \(\Longleftrightarrow\) \(\delta \leq 1\). Thus, if \(\delta \leq 1\), (NRA) and (NGA) hold, while a financial bubble exists if \(\delta > 1\). For this particular example, we see that Condition \ref{cond: E3} is sharp.
\end{example}
\section{Comments on Related Literature}\label{sec:LMU}
We comment on related literature.
Let us start with the article of \cite{Mijatovic2012}, which has also motivated our research.

The financial market is assumed to include one risky asset \(S = (S_t)_{t \in [0, T]}\), which is modeled as a one-dimensional homogeneous diffusion with dynamics
\begin{align}\label{SDE:MU}
\dd S_t = \mu(S_t)\dd t + \sigma(S_t)\dd B_t, \quad S_0 > 0,
\end{align}
where \(B = (B_t)_{t \in [0, T]}\) is a one-dimensional Brownian motion.
In our framework we do not model the asset price process as a homogeneous diffusion, but suppose that its log returns are time-inhomogeneous diffusions. 

Mijatovi\'c and Urusov assume that \(\sigma \colon (0, \infty) \to \mathbb{R} \backslash \{0\}\) and \(\mu \colon (0, \infty) \to \mathbb{R}\) are Borel functions which satisfy the Engelbert-Schmidt conditions, i.e.
\begin{align}
\frac{1 + |\mu|}{\sigma^2} \in L^1_\textup{loc}((0, \infty)).
\end{align}
In this case, the SDE \eqref{SDE:MU} has a solution (up to explosion) which is unique in law.
For simplicity, let us suppose that the coefficients \(\sigma\) and \(\mu\) are chosen such that almost all paths of \(S\) are positive. For explicit conditions on \(\mu\) and \(\sigma\) we refer to Feller's test, see Theorem 5.5.29 in \cite{KaraShre} or Theorem \ref{theo: Feller test} in Appendix \ref{Appendix}.
Mijatovi\'c and Urusov furthermore discuss the case where \(Y\) can explode to 0, as well as the infinite time horizon.                                               
The notions of E(L)MMs and S(L)MDs are defined similar as in our setting.
Here, a SLMD exists if and only if 
\begin{align}\label{eq: MU MPR}
\frac{\mu^2}{\sigma^4} \in L^1_\textup{loc} ((0, \infty)).
\end{align}
By Corollary 3.4 in \cite{Mijatovic2012}, an ELMM exists if and only if \eqref{eq: MU MPR} holds and 
\begin{align}\label{eq: SDE MU1}
\int_0^1 \frac{x}{\sigma^2 (x)} \dd x = \infty.
\end{align}
Furthermore, by Theorem 3.6 in \cite{Mijatovic2012}, an EMM exists if and only if \eqref{eq: MU MPR} and \eqref{eq: SDE MU1} hold and
\begin{align}\label{eq: MU SDE 2}
\int_1^\infty \frac{x}{\sigma^2(x)} \dd x = \infty.
\end{align}
By Theorem 3.11 in \cite{Mijatovic2012}, a SMD exists if and only if \eqref{eq: MU MPR} and \eqref{eq: MU SDE 2} hold.
These results show that 
\begin{align}
\textup{(NFLVR) } + \textup{ (NRA)} \quad \Longleftrightarrow\quad \textup{(NGA)}.
\end{align}
This differs slightly from the one-dimensional version of our setting, where (NFLVR) always holds and 
\begin{align}
\textup{(NRA)} \quad \Longleftrightarrow\quad \textup{(NGA)},
\end{align}
see Section \ref{section 4}. The condition \eqref{eq: MU SDE 2} is comparable with our conditions for the existence and non-existence of an ELMM as given in Proposition \ref{dim coro}. As we showed in Proposition \ref{proposition: 1D}, in our setting \eqref{eq: MU SDE 2} does not suffices to guarantee the existence of an EMM.

Now, considering the one-dimensional setting, it is natural to ask what can be said in multi-dimensional cases. More precisely, how the dimension, i.e. the number of assets in the market, affects the existence and absence of arbitrage and whether one can give deterministic conditions which are comparable to those in the one-dimensional case. We have given such conditions in the Theorems \ref{theo: LG} and \ref{theo: det NE 1}.

A classical example where low- and high-dimensional behavior of processes differs is the recurrence and transience of Brownian motion. It is well-known that Brownian motion is recurrent in dimensions one and two and transient otherwise. 
Related to this fact, in Section \ref{section 4} we constructed a type of financial market which is free of arbitrage in dimensions one and two and allows arbitrage otherwise.

Recall that Mijatovi\'c and Urusov assume that there is only one choice of the real-world measure and one candidate for an ELMM, the law of the solution process to the SDE \eqref{eq: SDE MU1}. It is an interesting question what can be said if we forgo on uniqueness.
Our results also apply if uniqueness does not hold. Indeed, we saw that it is often possible to find an E(L)MM under the candidates. But, we also saw that not necessarily all candidates are E(L)MMs and that the ELMM may depend on the choice of the real-world measure.

The absence and existence of arbitrage was also studied by \cite{MAFI:MAFI530}.
For simplicity, we only describe his setting without interest rate and information process. 

The market includes one risky asset \(S = (S_t)_{t \in [0, T]}\) given as the stochastic exponential
\begin{align}
\dd S_t = S_t \dd \tilde{X}_t,\quad S_0 = 1,
\end{align}
where \(\tilde{X} = (\tilde{X}_t)_{t \in [0, T]}\) is a continuous It\^o process, the so-called \emph{aggregated excess-returns process}.
In his article, Lyasoff presents the one-dimensional case in full detail and remarks that the multi-dimensional case can be treated along the lines.
It is assumed that
\begin{align}\label{Ly MPR cond}
\dd \tilde{X}_t = \alpha (t, \tilde{X}) \theta_t \dd t+\alpha (t, \tilde{X}) \dd B_t,\quad \tilde{X}_0 = 0,
\end{align}
where \(\alpha\) and \(\theta\) satisfy minimal conditions such that \(\tilde{X}\) is well-defined and \(B = (B_t)_{t \in [0, T]}\) is a one-dimensional Brownian motion.

Denote \(\mu\) the law of \(\tilde{X}\) seen as a probability measure on \((\Omega, \mathcal{F})\).
In this setting, a probability measure \(Q\) on \((\Omega, \mathcal{F})\) is called an ELMM if \(Q \sim \mu\) and \(X\) is a local \((\F, Q)\)-martingale.

It is presumed that for \(\dd t \otimes \mu\)-a.a. \((t,\omega) \in [0, T] \times \Omega\) it holds that \(\alpha (t, \omega) > 0\) and a.s.
\begin{align}\label{Ly: Theta SI}
\int_0^T \|\theta_s\|^2 \dd s < \infty.
\end{align}
This corresponds to \eqref{eq: c inte cond} and \eqref{eq: MPRE} in our setup.

By Proposition 2.3 in \cite{MAFI:MAFI530}, the existence of an ELMM is equivalent to \(\mathscr{W} \sim \nu\), where \(\mathscr{W}\) is the (one-dimensional) Wiener measure (on the time interval \([0, T]\)) and \(\nu\) is the law of the process \begin{align}\label{theta}
\check{X}_t\triangleq \int_0^t \alpha^{-1} (s, \tilde{X}) \dd \tilde{X}_s = \int_0^t \theta_s \dd s + B_t,\quad t \in [0, T],\end{align} which is called the \emph{normalized excess-return process}.

In general, this type of condition is very different from the deterministic type we gave in this article.

The one-dimensional version of our setting can partially be included into the setting of Lyasoff. That means, if \(d = m = 1\) and \(a^\frac{1}{2}\) is invertible, then we can choose \(\alpha(t, \tilde{X}) \triangleq a^\frac{1}{2}(t, \tilde{X}_t)\) and \(\theta_t \triangleq a^{-\frac{1}{2}}(t, \tilde{X}_t) b(t, \tilde{X}_t)\).

\appendix
\section{Technical Proofs}
\subsection{Proof of Proposition \ref{prop: uni 1d}}\label{App: pf 1}
As in the proof of Theorem \ref{theo: LG}, it suffices to show that the SDE
\begin{align}\label{eq: 1D SDE Le Gall}
\dd Y_t = a^\frac{1}{2}(t, Y_t) \dd W_t, \quad Y_0 = x_0, 
\end{align}
satisfies pathwise uniqueness. We localize a classical argument due to Le Gall. 
Let \(Y = (Y_t)_{t \in [0, T]}\) and \(U = (U_t)_{t \in [0, T]}\) be two solution processes of the SDE \eqref{eq: 1D SDE Le Gall} for the same driving system, see Appendix \ref{Appendix} for this terminology. 
Set \(\xi_n \triangleq \inf(t \in [0, T] \colon \|Y_t\| \vee \|U_t\| \geq n).\)
For all \(t \in [0, T]\), we have 
\begin{equation}\begin{split}
\int_0^{t \wedge \xi_n} &\frac{1}{h_n (Y_s - U_s)} \1_{\{Y_s - U_s > 0\}} \dd \hspace{0.05cm}\lle Y - U\rre_s 
\\&= \int_0^{t \wedge \xi_n} \frac{\|a^\frac{1}{2} (s, Y_s) - a^\frac{1}{2} (s, U_s)\|^2}{h_n(\|Y_s -  U_s\|)} \1_{\{Y_s - U_s > 0\}} \dd s 
\\&\leq \int_0^T \zeta_n (s) \dd s < \infty.
\end{split}\end{equation}
Thus, by a lemma of Le Gall, see Lemma \ref{lem: Le Gall} in Appendix \ref{Appendix: Le Gall und co}, the local time of \((Y_{t \wedge \xi_n} - U_{t \wedge \xi_n})_{t \in [0, T]}\) in the origin a.s. vanishes. Hence, by Tanaka's formula, see, for instance, Theorem VI.1.2 in \cite{RY}, for all \(t \in [0, T]\) it holds that 
\begin{align}\label{eq: 1D Tanaka}
\|Y_{t \wedge \xi_n} - U_{t \wedge \xi_n}\| = \int_0^{t \wedge \xi_n}\textup{sgn}(Y_s - U_s) \dd Y_s -  \int_0^{t \wedge \xi_n} \textup{sgn}(Y_s - U_s) \dd U_s.
\end{align}
We deduce from the bound
\begin{equation}\begin{split}
\lle Y - U\rre_{T \wedge \xi_n} &= \int_0^{T \wedge \xi_n} \|a^\frac{1}{2} (s, Y_s) - a^\frac{1}{2} (s, U_s)\|^2 \dd s \\&\leq h_n(2n) \int_0^T \zeta_n (s) \dd s < \infty
\end{split}
\end{equation}
and \eqref{eq: 1D Tanaka} that \((\|Y_{t \wedge \xi_n} - U_{t \wedge \xi_n}\|)_{t \in [0, T]}\) is a martingale starting at the origin. Thus, taking expectation, by continuity, a.s. \(Y_{t \wedge \xi_n} = U_{t \wedge \xi_n}\) for all \(t \in [0, T]\) and \(n \in \mathbb{N}\). Since \(\xi_n \nearrow \infty\) as \(n \to \infty\), we conclude that the SDE \eqref{eq: 1D SDE Le Gall} satisfies pathwise uniqueness.
\qed
\subsection{Proof of Lemma \ref{lem: cm no explosion}}\label{App: pf 2}
We start with the observation that \(Q_n\) solves a stopped MP as defined in the following
\begin{definition}\label{def: stopped MP}
	For an \(\F\)-stopping time \(\tau\) we say that a probability measure \(P\) on \((\Omega, \mathcal{F})\) solves the (stopped) MP \((b, a, x; \tau)\) if the stopped process \((M^f_{t \wedge \tau, x})_{t \in [0, T]}\) (see \eqref{Mf} for this notation) is a local \((\F, P)\)-martingale for all \(f \in C^2(\mathbb{R}^d)\). %We denote the set of solutions by \(\mathcal{M}(b, a, x; \tau)\).
\end{definition}
\begin{lemma} \label{lem: Qn}
	The probability measure \(Q_n\) solves the MP \(((0, \mu), a, x_0; \tau_n)\). 
\end{lemma}
\begin{proof}
	By Girsanov's theorem, see Theorem \ref{theo: Gir} in Appendix \ref{rem: SMP}, \eqref{eq: MPRE} and \eqref{eq: set M}, the stopped coordinate process \((X_{t \wedge \tau_n})_{t \in [0, T]}\) is an \((\F, Q_n)\)-semimartingale with decomposition
	\begin{align}
	X_{t \wedge \tau_n} = \int_0^{t \wedge \tau_n} (0, \mu)(s, X_s)\dd s + X^c_{t \wedge \tau_n},\quad t \in [0, T], 
	\end{align}
	where \((X^c_{t \wedge \tau_n})_{t \in [0, T]}\) is a local \((\F, Q_n)\)-martingale with quadratic variation process \begin{align}\left(\int_0^{t \wedge \tau_n} a(s, X_s)\dd s\right)_{t \in [0, T]}.\end{align}
	For \(f \in C^2(\mathbb{R}^d)\), by It\^o's formula, the process \(M^f_{\cdot \wedge \tau_n, x_0}\), defined as in \eqref{Mf} with \(b\) replaced by \((0, \mu)\), is a local \((\F, Q_n)\)-martingale. This implies the claim.
\end{proof}

First, we assume that Condition \ref{cond: EL1} holds. 
We use ideas given in the proofs of the Theorems 10.2.1 and 10.2.3 in \cite{SV}.
Define a sequence \((u_n)_{n \in \mathbb{N}}\) of Borel functions \([r, \infty) \to [1, \infty)\) recursively by setting \(u_0 \triangleq 1\) and 
\begin{align}
u_n (x) = \int_r^x \exp\left(- \int_r^y B(z)\dd z\right)\int_r^y u_{n-1}(v) \frac{2 \exp \left(\int_r^v B(z)\dd z\right)}{A(v)} \dd v \dd y
\end{align}
for \(x \in [r, \infty).\)
Since \(A\) and \(B\) are continuous, we obtain
\begin{align}\label{eq: ODE n}
A(x) u''_n(x) + A(x) B(x) u'_n(x) = 2 u_{n-1}(x),\quad x \in [r, \infty).
\end{align}
Moreover, by induction, it follows that 
\begin{align}\label{eq: bound u}
u_n (x) \leq \frac{[u_1(x)]^n}{n!},\quad n \in \mathbb{N} \cup \{0\}, x \in [r, \infty).
\end{align}
In particular, since
\begin{align}\label{eq: bound u'}
u'_n (x) = \exp \left(- \int_r^x B(z) \dd z\right) \int_r^x u_{k-1}(y) \frac{2 \exp \left(\int_r^y B(z)\dd z\right)}{A(v)} \dd y, 
\end{align}
it holds that 
\begin{align}
u'_n(x) \leq u'_1 (x) \frac{[u_1(x)]^{n-1}}{(n-1)!},\quad n \in \mathbb{N}, x \in [r, \infty).
\end{align}
Thus, since \(u_n \geq 0\) and \(u'_n \geq 0\) for all \(n \in \mathbb{N} \cup \{0\}\), the sums \(\sum_{n = 0}^\infty u_n\) and \(\sum_{n= 0}^\infty u_n'\) converge absolutely uniformly on compact subsets of \([r, \infty)\). Solving \eqref{eq: ODE n} for \(u''_n\) and using again the bounds \eqref{eq: bound u} and \eqref{eq: bound u'}, it follows that also \(\sum_{n = 1}^\infty u''_n\) converges absolutely uniformly on compact subsets of \([r, \infty)\).
Summarizing this, \(u \triangleq \sum_{n = 0}^\infty u_n\) is a twice continuously differentiable increasing function with \(u' = \sum_{n = 0}^\infty u'_n\) and \(u'' = \sum_{n = 1}^\infty u''_n\). In particular, using \eqref{eq: ODE n}, it holds that
\begin{align}\label{u ODE}
A(x) u''(x) + A(x) B(x) u'(x)= 2 u(x),\quad  x \in [r, \infty).
\end{align}
By \eqref{cond: eq H}, we have
\begin{align}\label{eq: u expl}
u(x) \geq 1 + u_1(x) \xrightarrow{\quad x \to \infty\quad} \infty.
\end{align}
Let \(\phi \in C^2 (\mathbb{R}^d)\) be such that \(\phi \geq 1\) and \(\phi(x) \triangleq u\big(\frac{\|x\|^2}{2}\big)\) for \(x \in \{z \in \mathbb{R}^d \colon \|z\| > \sqrt{2 r}\}\). We have for \(\dd t\)-a.a. \(t \in [0, T]\) and all \(x \in \{z \in \mathbb{R}^d\colon \|z\| \leq \sqrt{2 r}\}\)
\begin{equation}\begin{split}
\mathcal{K}^{((0, \mu), a)}_t \phi (x) &= \la (0, \mu)(t, x),  \nabla \phi (x)\ra + \frac{1}{2}\textup{trace} \left( a(t, x) \nabla^2 \phi(x) \right) 
\\&\leq  \|\mu(t, x)\| \|\nabla \phi(x)\| + \frac{1}{2} \|a(t, x)\| \|\nabla^2 \phi (x)\|
\\&\leq\zeta (t) \max \left( \sup_{\|y\| \leq \sqrt{2 r}} \|\nabla \phi(y)\|, \sup_{\|y\| \leq \sqrt{2 r}}\frac{ \|\nabla^2 \phi(y)\|}{2} \right)\triangleq \zeta (t) C,
\end{split}
\end{equation}
where \(\zeta\) is as in Condition \ref{cond: EL1}.
\begin{lemma}\label{lem: bound phi}
	For \(\dd t\)-a.a. \(t \in [0, T]\) and all \(x \in \mathbb{R}^d\) it holds that \begin{align}\mathcal{K}^{((0, \mu), a)}_t \phi (x) \leq \zeta (t) \max (C, 1) \phi(x).\end{align}
\end{lemma}
\begin{proof}
	It suffices to prove the claim for \(\dd t\)-a.a. \(t \in [0, T]\) and all \(x \in \{z \in \mathbb{R}^d\colon \|z\|> \sqrt{2 r}\}
	\). By \eqref{cond bound 1}, \eqref{cond bound 2} and \eqref{u ODE}, we obtain
	\begin{equation}
	\begin{split}
	\mathcal{K}^{((0, \mu), a)}_t \phi (x) 
	&= \left\la(0, \mu)(t, x), \nabla u \left(\frac{\|x\|^2}{2}\right) \right\ra + \frac{1}{2}\textup{trace} \left( a(t, x) \nabla^2 u \left(\frac{\|x\|^2}{2}\right) \right)
	\\&=  \frac{1}{2}\la x, a(t, x) x\ra u'' \left(\frac{\|x\|^2}{2} \right)  \\&\qquad\qquad+  \frac{1}{2}\left( \textup{trace } a(t, x) + 2\la x, (0, \mu)(t, x)\ra\right) u' \left(\frac{\|x\|^2}{2}\right) 
	\\&\leq \frac{\zeta (t)}{2} \la x, a(t, x) x\ra \left( u'' \left(\frac{\|x\|^2}{2} \right)  +  B\left(\frac{\|x\|^2}{2}\right) u' \left(\frac{\|x\|^2}{2}\right) \right) 
	\\&\leq \frac{\zeta (t)}{2}A \left(\frac{\|x\|^2}{2}\right) \left( u'' \left(\frac{\|x\|^2}{2} \right)  +  B\left(\frac{\|x\|^2}{2}\right) u' \left(\frac{\|x\|^2}{2}\right) \right) 
	\\&= \zeta (t) u \left(\frac{\|x\|^2}{2}\right), %                         = \phi(x).
	\end{split}
	\end{equation}
	where we also used that \(u' \geq 0\) and \(u'' + B u' \geq 0\), see \eqref{u ODE}.
\end{proof}
Now, set \begin{align}
U^n_t \triangleq \exp \left(- \max (C, 1) \int_0^{t \wedge \tau_n} \zeta (s) \dd s\right),\quad t \in [0, T].
\end{align} 
By Lemma \ref{lem: Qn} and It\^o's formula,
\begin{equation}\begin{split}
U^n \phi(X_{\cdot \wedge \tau_n}) =  \phi(x_0) &+ \int_0^{\cdot \wedge \tau_n}  U^n_s \left(- \max (C, 1) \zeta(s) \phi(X_s) + \mathcal{K}^{((0, \mu), a)}_s \phi(X_s)\right) \dd s 
\\&+ \textup{local \((\F, Q_n)\)-martingale starting at the origin}.
\end{split} 
\end{equation}
Thus, by Lemma \ref{lem: bound phi}, we have 
\begin{align}
U^n \phi(X_{\cdot \wedge \tau_n}) \leq \textup{local \((\F, Q_n)\)-martingale starting at \(\phi(x_0)\)}. 
\end{align}
Since \(U^n_t \phi(X_{t \wedge \tau_n}) \geq 0\) for all \(t \in [0, T]\), the local \((\F, Q_n)\)-martingale on the r.h.s. is non-negative and thus an \((\F, Q_n)\)-supermartingale.
Hence, for \(n > \sqrt{2 r}\), we obtain
\begin{equation}\begin{split}
\phi(x_0) &\geq E^{Q_n} [U^n_T \phi(X_{T \wedge \tau_n})] \\&\geq E^{Q_n}[U^n_T \phi(X_{\tau_n}) \1_{\{\tau_n \leq T\}}]
%\\&=E^{Q_n} \left[ \exp \left(- \max (C, 1) \int_0^{\tau_n} \zeta (s) \dd s\right) \phi (X_{\tau_n}) \1_{\{\tau_n \leq T\}}\right] 
\\&=  E^{Q_n} \left[ \exp \left(- \max (C, 1) \int_0^{ \tau_n} \zeta (s) \dd s\right) u \left(\frac{n^2}{2}\right) \1_{\{\tau_n \leq T\}}\right] 
\\&\geq   \exp \left(- \max (C, 1) \int_0^{T} \zeta (s) \dd s\right) u \left(\frac{n^2}{2}\right) Q_n(\tau_n \leq T).
\end{split}
\end{equation}
By \eqref{eq: u expl}, this bound yields that \(\lim_{n \to \infty} Q_n (\tau_n \leq T) = 0\).

Now, assume that Condition \ref{cond: EL3} holds. 
The argument is very similar to the previous one.
Set \(\psi(x) \triangleq 1 + \|x\|^2\). 
By Condition \ref{cond: EL3}, for \(\dd t\)-a.a. \(t \in [0, T]\) and all \(x \in \mathbb{R}^d\)
\begin{align}\label{eq: psi bound}
\mathcal{K}^{((0, \mu), a)}_t \psi(x) = 2 \la b(t, x), x\ra + \textup{trace } a(t, x) \leq \zeta (t) (1 + \|x\|^2) = \zeta (t) \psi(x).
\end{align}
Now, set \begin{align}
V^n_{t} \triangleq \exp \left( - \int_0^{t \wedge \tau_n} \zeta (s)\dd s \right) \psi (X_{t \wedge \tau_n}),\end{align} and note, as above, that by It\^o's formula and \eqref{eq: psi bound}
\begin{align}
V^n \leq \textup{\((\F, Q_n)\)-supermartingale starting at \(\psi(x_0)\)}.
\end{align}
Thus, we obtain
\begin{align}
\psi(x_0) \geq E^{Q_n} [V^n_T] \geq \exp \left(- \int_0^T \zeta (s)\dd s \right) \left(1 + n^2\right) Q_n(\tau_n \leq T).
\end{align}
This proves that \(\lim_{n \to \infty} Q_n(\tau_n \leq T) = 0\).

Finally, we assume that Condition \ref{cond: EL MCKEAN} holds. We use ideas of \cite{mckean1969stochastic}.
For \(n \geq z_0\), set 
\begin{align}\label{eq: first time alpha}
\alpha_n \triangleq \frac{1}{2 n^2} \left( \sqrt{d^2 + \frac{4 n^2}{\gamma(n)}} - d\right),
\end{align}
and for \(t \in [0, T]\)
\begin{equation}
\begin{split}
\widehat{Z}_t &\triangleq \exp \left( \alpha_n \int_0^{t \wedge \tau_n} \la X_s, \dd X_s\ra - \frac{\alpha^2_n}{2} \int_0^{t \wedge \tau_n} \la X_s, a(s, X_s) X_s \ra \dd s\right) 
\\&= \exp \left( \frac{\alpha_n}{2} \left(\|X_{t \wedge \tau_n}\|^2 - \|x_0\|^2\right) \right)\\ &\qquad\quad\times \exp\left(- \frac{\alpha_n}{2} \int_0^{t \wedge \tau_n}\left( \textup{trace } a(s, X_s) + \alpha_n \la X_s, a(s, X_s) X_s \ra\right) \dd s\right),
\end{split}
\end{equation}
where the equality follows from Lemma \ref{lem: Qn} and It\^o's formula and holds up to a \(Q_n\)-null set. Recall that we assume that \(\mu = 0\), i.e. that \((X_{t \wedge \tau_n})_{t \in [0, T]}\) is an \((\F, Q_n)\)-martingale.
Hence, \(\widehat{Z} = (\widehat{Z}_t)_{t \in [0, T]}\) is a non-negative local \((\F, Q_n)\)-martingale and therefore an \((\F, Q_n)\)-supermartingale. 
Thus, for \(n > z_0\)
\begin{equation}
\begin{split}
1 &\geq E^{Q_n} \left[\widehat{Z}_{T \wedge \tau_n} \right] \\&\geq E^{Q_n} \left[\widehat{Z}_{\tau_n} \1_{\{\tau_n \leq T\}}\right]
\\&\geq \exp \left( \frac{\alpha_n}{2} \left(n^2 - \|x_0\|^2\right) \right)E^{Q_n}\left[ \exp \left(- \frac{\alpha_n \tau_n \gamma (n)}{2}\left(d + \alpha_n n^2 \right)\right) \1_{\{\tau_n \leq T\}} \right]
\\&= \exp \left( \frac{\alpha_n}{2} \left(n^2 - \|x_0\|^2\right) \right)E^{Q_n}\left[ \exp \left(- \frac{\tau_n}{2}\right) \1_{\{\tau_n \leq T\}} \right]
\\&\geq \exp \left( \frac{\alpha_n}{2} \left(n^2 - \|x_0\|^2\right) - \frac{T}{2} \right) Q_n(\tau_n \leq T),
\end{split}
\end{equation}
where we use that \eqref{eq: first time alpha} is chosen such that \(\alpha_n \gamma(n) (d + \alpha_n n^2) = 1\).
If \(\limsup_{n \to \infty} \frac{n^2}{\gamma (n)} = \infty\), then
\begin{equation}\begin{split}
\limsup_{n \to \infty} &\frac{\alpha_n}{2} \left(n^2 - \|x_0\|^2\right) 
\\&= \limsup_{n \to \infty} \frac{1}{4} \left(\sqrt{d^2 + \frac{4 n^2}{\gamma(n)}} - d\right) \cdot \left(\frac{n^2 - \|x_0\|^2}{n^2}\right)= \infty 
\end{split}
\end{equation}
and we can conclude that \(\lim_{n \to \infty} Q_n (\tau_n \leq T) = 0\).

In the remaining proof we assume the second part of Condition \ref{cond: EL MCKEAN}. W.l.o.g. we can additionally assume that \begin{align}\label{eq: add assp} \limsup_{n \to \infty} \frac{n^2}{\gamma(n)} < \infty,\end{align}
since otherwise we are in the case discussed above.
With a little abuse of notation, we redefine \(\alpha_n\), see \eqref{eq: first time alpha}, by setting \begin{align}\alpha_n \triangleq \frac{1}{\gamma(n)}.\end{align} 
Furthermore, we set 
\begin{align}
c_n \triangleq d + \frac{n^2}{\gamma (n)}.
\end{align}
By the optional stopping theorem, for all \(\epsilon > 0\) such that \(n - \epsilon \geq z_0\),  up to a \(Q_n\)-null set, we obtain
\begin{equation}\label{eq: help McKean}
\begin{split}
1 &\geq E^{Q_n} \left[\frac{\widehat{Z}_{T \wedge \tau_{n}}}{\widehat{Z}_{T \wedge \tau_{n - \epsilon}}} \bigg|\mathcal{F}_{\tau_{n - \epsilon}}\right] 
\\&\geq E^{Q_n} \left[\frac{\widehat{Z}_{\tau_{n}}}{\widehat{Z}_{\tau_{n - \epsilon}}} \1_{\{\tau_n \leq T\}} \bigg|\mathcal{F}_{\tau_{n - \epsilon}}\right] 
\\&\geq \exp \left(\frac{n^2 - (n - \epsilon)^2}{2 \gamma (n)} \right)
E^{Q_n} \left[ \exp \left(-  \frac{c_n }{2}\cdot (\tau_{n} - \tau_{n - \epsilon})\right) \1_{\{\tau_{n} \leq T\}} \big| \mathcal{F}_{\tau_{n - \epsilon}} \right]
\\&\geq \exp \left(\frac{n \epsilon}{2 \gamma (n)} \right)
E^{Q_n} \left[ \exp \left(-  \frac{c_n }{2}\cdot (\tau_{n} - \tau_{n - \epsilon})\right)\1_{\{\tau_{n} \leq T\}} \big| \mathcal{F}_{\tau_{n - \epsilon}} \right].
\end{split}
\end{equation}
Now, set \(t_k^m \triangleq \frac{k (n - z_0)}{m}\). Then, for \(n > z_0\) and \(m \geq 2\), using \eqref{eq: help McKean}, we obtain
\begin{equation}\begin{split}
\exp &\left(\frac{-c_n T}{2}\right) Q_n (\tau_n \leq T) \\&\ \ \leq E^{Q_n} \left[ \exp \left( \frac{-c_n (\tau_n - \tau_{z_0})}{2} \right) \1_{\{\tau_n \leq T\}} \right] 
\\&\ \ = E^{Q_n} \bigg[ \exp \bigg( - \frac{c_n}{2} \sum_{k = 0}^{m-1} \left(\tau_{n - t^m_k} - \tau_{n - t^m_{k + 1}}\right) \bigg)  \1_{\{\tau_{n - t^m_1} \leq T\} \cap \{\tau_n \leq T\}}\bigg]
\\&\ \ = E^{Q_n} \bigg[ \exp \bigg( - \frac{c_n}{2} \sum_{k = 1}^{m-1} \left(\tau_{n - t^m_k} - \tau_{n - t^m_{k + 1}}\right) \bigg) \1_{\{\tau_{n - t^m_1} \leq T\}} \\&\qquad\qquad\quad \times E^{Q_n} \left[ \exp \left(- \frac{c_n}{2}(\tau_{n} - \tau_{n - t^m_1})\right)\1_{\{\tau_{n} \leq T\}} \big| \mathcal{F}_{\tau_{n - t^m_1}}\right] \bigg]
\\& \ \ \leq E^{Q_n} \bigg[ \exp \bigg( - \frac{c_n}{2} \sum_{k = 1}^{m-1} \left(\tau_{n - t^m_k} - \tau_{n - t^m_{k + 1}}\right) \bigg)\1_{\{\tau_{n - t^m_1} \leq T\}} \\&\qquad\qquad\quad \times \exp \bigg( \frac{- (n - t^m_0) (t^m_1 - t^m_0)}{2\gamma(n - t^m_0)} \bigg) \bigg]
\\&\ \ \leq \exp \bigg( -\sum_{k = 0}^{m-1} \frac{ (n - t^m_k) ( t^m_{k + 1} - t^m_k)}{2\gamma(n - t^m_k)} \bigg)
\\&\ \ \leq \exp \bigg( -\sum_{k = 0}^{m-1} \frac{ (n - t^m_k) (t^m_{k + 1} - t^m_k)}{2\xi(n - t^m_k)} \bigg).
\end{split}
\end{equation}
Thus, letting \(m \to \infty\), for \(n > z_0\), we have
\begin{align}
Q_n (\tau_n \leq T) \leq \exp \left(\frac{c_nT}{2}\right) \exp \left( - \frac{1}{2}  \int_{z_0}^n \frac{z}{\xi(z)} \dd z \right).
\end{align}
Using \eqref{eq: cond MCKEAN} and \eqref{eq: add assp}, this implies that \(\lim_{n \to \infty} Q_n (\tau_n \leq T) = 0\) and the proof is complete. 
\qed
\subsection{Proof of Lemma \ref{lem: techn uni}}\label{App: pf 3}
Suppose that Condition \ref{cond: U1} holds. Let \(\Phi_n \in C^2(\mathbb{R}^d)\) such that \(\Phi_n\) has compact support, \(\Phi_n \in [0, 1]\) and \(\Phi_n = 1\) on \(\{x \in \mathbb{R}^d \colon \|x\| \leq n\}\). Then, set \(a_n (t, x) \triangleq \Phi_n (x) a(t, x) + (1- \Phi_n(x)) \textup{Id}\), where \(\textup{Id}\) denotes the identity matrix. Let \(Q_1\) and \(Q_2\) be solutions to the MP \((0, a, x_0)\).
By the Theorems 7.2.1 and 10.1.1 in \cite{SV}, there exists a solution to the MP \((0, a_n, x_0)\) which coincides with \(Q_1\) and \(Q_2\) on \(\mathcal{F}_{\tau_n} = \sigma(X_{t \wedge \tau_n}, t \in [0, T])\).
Here, Theorem 10.1.1 in \cite{SV} is comparable to Proposition \ref{prop: loc uni} in Appendix \ref{Appendix} for time-inhomogeneous MPs.
Since \(\tau_n (\omega) \nearrow \infty\) as \(n \to \infty\) for all \(\omega \in \Omega\), we have for all \(G \in \mathcal{F}\)
\begin{align}
Q_1(G) = \lim_{n \to \infty} Q_1(G \cap \{\tau_n > T\}) = \lim_{n \to \infty} Q_2(G \cap \{\tau_n > T\}) = Q_2(G).
\end{align}
Thus, we have proven uniqueness.

Suppose that Condition \ref{cond: U2} holds. We use a classical argument based on Gronwall's lemma to show that the SDE
\begin{align}\label{eq: SDE a}
\dd Y_t = a^\frac{1}{2} (t, Y_t) \dd W_t, \quad Y_0 = x_0, 
\end{align}
where \(W\) is a \(d\)-dimensional Brownian motion, satisfies pathwise uniqueness. Then, by the Yamada-Watanabe theorem, see Proposition \ref{theo: YW} in Appendix \ref{Appendix}, the SDE satisfies also uniqueness in law and, by Theorem \ref{theo: equivalence MP SDE up to expl} in Appendix \ref{Appendix}, there exists at most one solution to the MP \((0, a, x_0)\).
Let \(Y = (Y_t)_{t \in [0, T]}\) and \(U = (U_t)_{t \in [0, T]}\) be solution processes to the SDE \eqref{eq: SDE a} on the same driving system w.r.t. the same Brownian motion \(W = (W_t)_{t \in [0, T]}\). Set \(\rho_n \triangleq \inf(t \in [0, T] \colon \|Y_t\| \vee \|U_t\| \geq n)\). By It\^o's isometry and Fubini's theorem we obtain for all \(t \in [0, T]\) 
\begin{equation}\begin{split}
E\left[ \|Y_{t \wedge \rho_n} - U_{t \wedge \rho_n}\|^2\right] &= E \left[ \int_0^{t \wedge \rho_n} \|a^\frac{1}{2} (s, Y_s) - a^\frac{1}{2} (s, U_s)\|^2 \dd s \right] \\&\leq \int_0^{t} \zeta_n (s) E\left[ \|Y_{s \wedge \rho_n} - U_{s \wedge \rho_n}\|^2 \right]\dd s.
\end{split}
\end{equation}
Therefore, by Lemma \ref{lem: Gronwall} in Appendix \ref{Appendix: Le Gall und co} and the continuity of \((Y_{t \wedge \rho_n} - U_{t \wedge \rho_n})_{t \in [0, T]}\), a.s. \(Y_{t \wedge \rho_n} = U_{t \wedge \rho_n}\) for all \(t \in [0, T]\) and \(n \in \mathbb{N}\). Since \(\rho_n \nearrow \infty\) as \(n \to \infty\), this yields that the SDE \eqref{eq: SDE a} satisfies pathwise uniqueness.
\qed
\subsection{Proof of Lemma \ref{lem: Z compactification}}\label{App: pf 4}
First, we prove a pathwise uniqueness result for the SDE \eqref{SDE: Z}. We use an argument as in the proof of Proposition \ref{prop: uni 1d}.
Suppose that \(U = (U_t)_{t \in [0, \infty)}\) and \(V = (V_t)_{t \in [0, \infty)}\) are two continuous \(I\)-valued processes such that 
\begin{align}
\dd U_t &= A^{\frac{1}{2}} (U_t) \1_{\{t \leq \gamma_n(U)\}} \dd W_t + \frac{1}{2} A(U_t) B(U_t)\1_{\{ t \leq \gamma_n(U)\}} \dd t,\quad U_0 = \frac{1}{2},\\
\dd V_t &= A^{\frac{1}{2}} (V_t) \1_{\{ t \leq \gamma_n(V)\}} \dd W_t + \frac{1}{2} A(V_t) B(V_t)\1_{\{ t \leq \gamma_n(V)\}} \dd t,\hspace{0.12cm}\quad V_0 = \frac{1}{2},
\end{align}
where \(W = (W_t)_{t \in [0, \infty)}\) is a one-dimensional Brownian motion.
Set \begin{align}
Y_t \triangleq U_{t \wedge \gamma_n(U) \wedge \gamma_n(V)} - V_{t \wedge \gamma_n (U) \wedge \gamma_n(V)}.
\end{align}
Let \(\rho_n\) be as in Condition \ref{cond: NL1}. 
By \eqref{gamma n} we have for all \(t \in [0, \infty)\)
\begin{equation}\label{rho bound}
\begin{split}
\int_0^t \frac{1}{\rho_n(Y_s)} &\1_{\{Y_s > 0\}} \dd\hspace{0.05cm} \lle Y\rre_s 
\\&= \int_0^t \frac{\|A^{\frac{1}{2}}(U_s) - A^{\frac{1}{2}} (V_s)\|^2}{\rho_n(\|Y_s\|)} \1_{\{Y_s > 0\}} \1_{\{s \leq \gamma_n(U) \wedge \gamma_n(V)\}} \dd s
\\&\leq \int_0^t \frac{\rho_n(\|Y_s\|)}{\rho_n(\|Y_s\|)} \1_{\{Y_s > 0\}} \1_{\{s \leq \gamma_n(U) \wedge \gamma_n(V)\}} \dd s
\leq t.
\end{split}
\end{equation}
Using \eqref{rho cond} and a lemma by Le Gall, see Lemma \ref{lem: Le Gall} in Appendix \ref{Appendix: Le Gall und co}, we obtain that the local time of \(Y\) in the origin a.s. vanishes. Hence, by Tanaka's formula, it holds that 
\begin{align}
\|Y\| = \int_0^\cdot \textup{sgn}(Y_s) \dd Y_s.
\end{align}
Since, 
\begin{align}\label{brow mart}
\int_0^{t \wedge \gamma_n(U) \wedge \gamma_n(V)} \left(A^{\frac{1}{2}} (U_s) - A^{\frac{1}{2}}(V_s)\right)^2 \dd s \leq 2 t \sup_{x \in [\frac{1}{n}, n]} \|A(x)\| < \infty, 
\end{align}
the stochastic integral \(\int_0^{\cdot \wedge \gamma_n(U) \wedge \gamma_n(V)} \textup{sgn}(Y_s) (A^{\frac{1}{2}} (U_s) - A^{\frac{1}{2}} (V_s)) \dd W_s\) is a martingale. Hence, we obtain from \eqref{kappa n}, Fubini's theorem and Jensen's inequality
\begin{equation}\begin{split}
E \left[ \|Y_t\| \right] &= \frac{1}{2} E \left[\int_0^{t \wedge \gamma_n(U) \wedge \gamma_n(V)} (A(U_s) B(U_s) - A(V_s) B(V_s)) \dd s\right]
\\&\leq \frac{1}{2} E \left[ \int_0^t \kappa_n (\|Y_s\|) \dd s\right]
\\&= \frac{1}{2} \int_0^t E [\kappa_n (\|Y_s\|)] \dd s
\\&\leq \frac{1}{2} \int_0^t \kappa_n( E[\|Y_s\|]) \dd s.
\end{split}
\end{equation}
Now, due to the hypothesis \eqref{kappa cond}, a lemma of Bihari, see Lemma \ref{lem: Bihari} in Appendix \ref{Appendix: Le Gall und co}, yields \(E[\|Y_t\|] = 0\). Thus, by the continuity of \(Y\), a.s. 
\begin{align}\label{eq: U stopped mini}
U_{t \wedge \gamma_n(U) \wedge \gamma_n(V)} = V_{t \wedge \gamma_n(U) \wedge \gamma_n(V)} \text{ for all } t \in [0, \infty).
\end{align}
Note that, by Galmarino's test, see Lemma III.2.43 in \cite{JS}, we have the implication
\begin{align}\label{eq: U stopped}
U_{t} = V_t \text{ for all } t \leq \gamma_n(U) \quad \Longrightarrow \quad \gamma_n(U) = \gamma_n(V).
\end{align}
By symmetry, this gives the implication 
\begin{align}\label{eq: U stopped 2}
U_{t} = V_t \text{ for all } t \leq \gamma_n(U) \wedge \gamma_n(V) \quad \Longrightarrow \quad \gamma_n(U) = \gamma_n(V),
\end{align}
and we conclude that a.s. \eqref{eq: U stopped mini} implies a.s. \(U_{t \wedge \gamma_n(U)} = V_{t \wedge \gamma_n(V)}\) for all \(t \in [0, \infty)\). In other words, we have shown that the SDE \eqref{SDE: Z} satisfies pathwise uniqueness. In fact, this yields that the SDE 
\begin{align}\label{eq: SDE up to exp}
\dd Y_t = A^\frac{1}{2} (Y_t) \dd W_t + \frac{1}{2} A(Y_t) B(Y_t) \dd t,\quad Y_0 = \frac{1}{2}, 
\end{align}
satisfies pathwise uniqueness up to explosion, see Definition \ref{def: pu expl} in Appendix \ref{Appendix}. 

Moreover, by Theorem \ref{theo: Feller test} in Appendix \ref{Appendix}, the SDE \eqref{eq: SDE up to exp} has a weak solution up to explosion, see Definition \ref{def: SP} in Appendix \ref{Appendix}.
Therefore, since weak existence and pathwise uniqueness implies strong existence, see Theorem \ref{theo: YW 2} in Appendix \ref{Appendix}, the proof is complete.
\qed
\subsection{Proof of Lemma \ref{lem: comparison}}\label{App: pf 5}
We set  \(Y^{n} = (Y^n_t)_{t \in [0, T]} \triangleq (Z_{t \wedge \gamma_n(Z) \wedge \xi_n} - U_{t \wedge \gamma_n(Z) \wedge \xi_n})_{t \in [0, T]}\). Let \(t \in [0, T]\) be fixed. As in \eqref{rho bound}, we obtain
\begin{align}
\int_0^t \frac{1}{\rho_n(Y^{n}_s)} \1_{\{Y_s^{n} >0\}} \dd\hspace{0.06cm}\lle Y^n\rre_s 
&\leq t.
\end{align}
Thus, by a lemma of Le Gall, see Lemma \ref{lem: Le Gall} in Appendix \ref{Appendix: Le Gall und co}, a.s. the local time of \(Y^n\) in the origin vanishes and, by Tanaka's formula, we obtain
\begin{align}
\max (Y^{n}_t, 0) = \int_0^t \1_{\{Y^{n}_s > 0\}} \dd Y^{n}_s.
\end{align}
Similar to \eqref{brow mart} one verifies that the Brownian part of \(\int_0^\cdot \1_{\{Y^{n}_s > 0\}} \dd Y^{n}_s\) is a martingale. Now, Fubini's theorem, \eqref{B1} and \eqref{eq: 4.11}, \(\phi_s \leq T\) for all \(s \in [0, T]\) and Jensen's inequality yield
\begin{equation}\begin{split}
E &\left[ \max (Y^{n}_t, 0) \right] 
\\&\ \ = E \left[ \int_0^{t \wedge \gamma_n(Z) \wedge \xi_n} \frac{\1_{\{Z_s > U_s\}}}{2} \left(A(Z_s) B(Z_s) - \frac{  2A(U_s) u(\phi_s, Y_{\phi_s})}{\la a(\phi_s, Y_{\phi_s}) Y_{\phi_s}, Y_{\phi_s}\ra} \right)\dd s\right] 
\\&\ \ \leq E \left[ \int_0^{t \wedge \gamma_n(Z) \wedge \xi_n}\frac{\1_{\{Z_s > U_s\}}}{2} \left(A(Z_s) B(Z_s) - A(U_s) B(U_s) \right)\dd s\right] 
\\&\ \ \leq \frac{1}{2} \int_0^t \kappa_n(E \left[ \max( Y^{n}_s, 0)\right]) \dd s. 
\end{split}
\end{equation}
By Bihari's lemma, see Lemma \ref{lem: Bihari} in Appendix \ref{Appendix: Le Gall und co}, \(E[\max(Y^{n}_t, 0)] = 0\). Since \(t \mapsto \phi_t\) is continuous, the process \(Y^{n}\) has continuous paths and we conclude the claim. 
\qed
\section{Girsanov's Theorem and Kunita-Watanabe's Decomposition}\label{rem: SMP}
We can reformulate MPs in terms of classical semimartingale theory. For a proof see Theorem 13.55 in \cite{J79}.
\begin{proposition}\label{prop: eq MP}
	Let \(P\) be a probability measure on \((\Omega, \mathcal{F})\). The following are equivalent:
	\begin{enumerate} 
		\item[\textup{(i)}] \(P\) solves the MP \((b,a, x)\).
		\item[\textup{(ii)}] We have \(P(X_0 = x) = 1\) and the coordinate process \(X\) is a continuous \((\F,P)\)-semimartingale with canonical decomposition
		\begin{align}\label{eq: App decompositio}
		X = x + \int_0^\cdot b(s, X_s) \dd s + X^c,
		\end{align}
		where \(X^c = (X^c_t)_{t \in [0, T]}\) is a continuous local \((\F, P)\)-martingale with quadratic variation process \(\big(\int_0^t a(s, X_s)\dd s\big)_{t \in [0, T]}\) and \(X^c_0 = 0\). 
		This decomposition is unique up to a \(P\)-null set. 
	\end{enumerate}
\end{proposition}

Let us now state a version of Girsanov's theorem. For a proof see Proposition \ref{prop: eq MP} and Girsanov's theorem for semimartingales, i.e. Theorem III.3.24 in \cite{JS}.
\begin{theorem}\label{theo: Gir}
	Let \(P\) and \(Q\) be two probability measures on \((\Omega, \mathcal{F})\) such that \(Q \ll P\) with \(Z_t = \frac{\dd Q}{\dd P}\big|_{\mathcal{F}_t}\) for \(t \in [0, T]\). If \(P\) solves the MP \((b, a, x)\), then there exists an \(\F\)-predictable process \(c= (c_t)_{t \in [0, T]}\) such that \(Q\)-a.s.
	\begin{align}\label{eq: c inte cond 1}
	\int_0^T \la a(s, X_s) c_s, c_s\ra \dd s < \infty 
	\end{align}
	and	\(X\) is a continuous \((\F, Q)\)-semimartingale with decomposition
	\begin{align}\label{eq: Gir decompo}
	X = x + \int_0^t \left( b(s, X_s) + a(s, X_s) c_s \right)\dd s + X^c, 
	\end{align}
	where \(X^c\) is a continuous local \((\F, Q)\)-martingale with quadratic variation process \(\big(\int_0^t a(s, X_s)\dd s\big)_{t \in [0, T]}\). Moreover, \(c\) is as above if and only if it solves, up to a \(P\)-null set, the equations
	\begin{align}
	\lle\la X, e_i\ra, Z \rre_t = \int_0^t Z_{s-} \la e_i, a(s, X_s) c_s\ra \dd s,\quad i \in \{ 1, ..., d\}, t \in [0, T],
	\end{align}
	where \(\lle \cdot \rre\) is the quadratic variation process relative to \(P\).
\end{theorem}

Finally, we recall a version of the decomposition theorem of Kunita and Watanabe. For a proof see Proposition \ref{prop: eq MP} and Theorem III.4.11 in \cite{JS}.
\begin{theorem}\label{theo: deco KW}
	Let \(P\) be a solution to the MP \((b, a, x)\) and let \(X^c\) be as in \eqref{eq: App decompositio}.
	Then, for each local \((\F, P)\)-martingale \(Z = (Z_t)_{t \in [0, T]}\) with \(Z_0 = 0\) we find an \(\mathbb{R}^d\)-valued \(\F\)-predictable process \(c = (c_t)_{t \in [0, T]}\) such that \eqref{eq: c inte cond 1} holds \(P\)-a.s. and a local \((\F, P)\)-martingale \(N = (N_t)_{t \in [0, T]}\) with \(\lle N, \la e_i,X^c\ra\rre = 0\) for all \(i = 1, ..., d\) such that 
	\begin{align}
	Z_t = \int_0^t \la c_s, \dd X^c_s\ra + N_t,\quad t \in [0, T].
	\end{align}
\end{theorem}

\section{MPs and SDEs up to Explosion}\label{Appendix}
Let \(I\subseteq \mathbb{R}^d\) be  a domain. 
We denote \(I_\Delta \triangleq I\cup \{\Delta\}\) the one-point compactification of \(I\) and by \(C([0, \infty), I_\Delta)\) be the space of continuous functions \([0, \infty) \to I_\Delta\). It is well-known that \(I_\Delta\) is metrizable and thus a Polish space, see \cite{cohn13}. We note that \(C([0, \infty), I_\Delta)\) is a Polish space when equipped with the local uniform topology. Let \(X = (X_t)_{t \in [0, \infty)}\) be the coordinate process on \(C([0, \infty), I_\Delta)\). Then, the Borel \(\sigma\)-field
\(\mathcal{B}(C([0, \infty), I_\Delta))\) coincides with \(\sigma(X_s, s \in [0, \infty))\). Furthermore, we define \(\mathcal{B}(C([0, t], I_\Delta)) \triangleq \sigma(X_s, s \in [0, t])\) for \(t \in [0, \infty)\). % and \(\mathbf{B}\triangleq (\mathcal{B}(C([0, t], I_\Delta)))_{t \in [0, \infty)}\).

Let \(\mathbb{W}\) be the set of all \(\omega \in C([0, \infty), I_\Delta)\) such that \(\omega(t) = \Delta\) for all \(t \geq \gamma_I(\omega) \triangleq \inf(t \in [0, \infty)\colon \omega(t) \not \in I)\). Then, \(\mathbb{W}\) is a closed subspace of \(C([0, \infty), I_\Delta)\) and, hence, itself a Polish space. Let \(X = (X_t)_{t \in [0, \infty)}\) be the coordinate process on \(\mathbb{W}\) and set \(\mathcal{W} \triangleq \sigma (X_s, s \in [0, \infty))\) and \(\mathbf{W} \triangleq (\mathcal{W}_{t+})_{t \in [0, \infty)}\) with \(\mathcal{W}_t \triangleq \sigma(X_s, s \in [0, t])\).

Fix two Borel functions \(\mu \colon [0, \infty) \times I \to \mathbb{R}^d\) and \(\sigma \colon [0, \infty) \times I \to \mathbb{R}^d \otimes \mathbb{R}^d\), \(y \in I\)  % and \(\nu\) be a probability measure on \((\mathbb{R}^d, \mathcal{B}(\mathbb{R}^d))\). 
and a sequence \((I_n)_{n \in \mathbb{N}}\subset \mathbb{R}^d\) of bounded domains such that \(I_{n} \subset I_{n+1}\) for all \(n \in \mathbb{N}\) and \(\bigcup_{n \in \mathbb{N}} I_n = I\). We set 
\begin{align}
\gamma_n (\omega) \triangleq \inf\left(t \in [0, \infty)\colon \omega(t) \not \in I_n\right) \wedge n,\quad \omega \in \mathbb{W}.
\end{align}
\begin{definition}\label{def: SP}
	A triplet \(((\Sigma, \mathcal{A}, \mathbf{A}, P); W; Y)\) is called \emph{weak solution up to explosion} to \begin{align}\label{eq: Appendix SDE}
	\dd Y_t = \mu(t, Y_t) \dd t + \sigma(t, Y_t) \dd W_t,\quad Y_0= y , 
	\end{align} if \((\Sigma, \mathcal{A}, \mathbf{A}, P)\) is a filtered probability space satisfying the usual hypothesis which supports a \(d\)-dimensional Brownian motion \(W = (W_t)_{t \in [0, \infty)}\) and an \(\mathbf{A}\)-adapted process \(Y = (Y_t)_{t \in [0, \infty)}\) with paths in \(\mathbb{W}\), such that \(P (Y_0 = y) = 1\) and \(P\)-a.s. for all \(n \in \mathbb{N}\) and all \(t \in [0, \infty)\)
	\begin{align}
	Y_{t \wedge \gamma_n(Y)} = Y_0 + \int_0^{t \wedge \gamma_n(Y)} \mu(s, Y_s) \dd s + \int_0^{t \wedge \gamma_n(Y)} \sigma(s, Y_s) \dd W_s, 
	\end{align}
	where it is implicit that the integrals are well-defined, i.e. \(P\)-a.s. for all \(n \in \mathbb{N}\) and all \(t \in [0, \infty)\)
	\begin{align}
	\int_0^{t \wedge \gamma_n(Y)} \left(\|\mu(s, Y_s)\| +  \|\sigma (s, Y_s) \sigma^* (s, Y_s)\|\right) \dd s < \infty.
	\end{align}
	The tuple \(((\Sigma, \mathcal{A}, \mathbf{A}, P); W)\) is called a \emph{driving system} of the SDE \eqref{eq: Appendix SDE} and the process \(Y\) is called \emph{solution process up to explosion}. If \(P(\gamma_I(Y) = \infty) = 1\), we call the solution process non-explosive.
\end{definition}
\begin{definition}\label{def: pu expl}
	We say that the SDE \eqref{eq: Appendix SDE} satisfies \emph{pathwise uniqueness up to explosion}, if all solution processes up to explosion on the same driving system are indistinguishable. Moreover, we say the the SDE \eqref{eq: Appendix SDE} satisfies \emph{uniqueness in law up to explosion}, if all solution processes have the same law, seen as a probability measure on \((\mathbb{W}, \mathcal{W})\).
\end{definition}
The following theorem is due to Yamada-Watanabe. Its proof is similar as in the non-explosive case, see, for instance, Proposition 5.3.20 in \cite{KaraShre}.
\begin{proposition}\label{theo: YW}	
	Pathwise uniqueness up to explosion implies uniqueness in law up to explosion. 
\end{proposition}
As in the non-explosive setting, pathwise uniqueness and weak existence implies strong existence. The proof is similar to the proof of the non-explosive case, see, for instance, Section 5.3.D, pp. 308, in \cite{KaraShre}. 
The Wiener measure on \((C([0, \infty), \mathbb{R}^d), \mathcal{B}(C([0, \infty), \mathbb{R}^d)))\) is denoted by \(\mathscr{W}\).
\begin{theorem}\label{theo: YW 2}
	Suppose that the SDE \eqref{eq: Appendix SDE} satisfies pathwise uniqueness up to explosion and has a weak solution up to explosion. 
	Let \((\Sigma, \mathcal{A}, \mathbf{A}, P)\) be a filtered probability space which supports a \(d\)-dimensional Brownian motion \((W_t)_{t \in [0, \infty)}\). Denote by \(\overline{\mathcal{B}(C([0,t],\mathbb{R}^d))}\) the completion of \(\mathcal{B}(C([0,t],\mathbb{R}^d))\) w.r.t. \(\mathscr{W}\).
	Then there exists a Borel map \(h \colon C([0, \infty), \mathbb{R}^d))  \to \mathbb{W}\) such that for all \(t \in [0, \infty)\)
	\begin{align}
	h^{-1} (\mathcal{W}_t) \subseteq \overline{\mathcal{B}(C([0,t],\mathbb{R}^d))},
	\end{align}
	and the process
	\(
	h(W)
	\)
	is a solution process up to explosion to the SDE \eqref{eq: Appendix SDE}.
\end{theorem}
We also introduce a MP up to explosion.
\begin{definition}\label{def: MP expl}
	We call a probability measure \(P\) on \((\mathbb{W}, \mathcal{W})\) a \emph{solution up to explosion to the MP \((\mu, \sigma \sigma^*, y)\)} if for all \(n \in \mathbb{N}\) and \(f \in C^2(I)\) the process \(M^f_{\cdot \wedge \gamma_n, y}\), defined as in \eqref{Mf} with \(a\) replaced by \(\sigma \sigma^*\), is a local \((\mathbf{W}, P)\)-martingale.
\end{definition}
\begin{theorem}\label{theo: equivalence MP SDE up to expl}
	A weak solution up to explosion to the SDE \eqref{eq: Appendix SDE} exists if and only if the MP \((\mu, \sigma \sigma^*, y)\) has a solution up to explosion.
	Moreover, the law of the solution process is a solution to the MP \((\mu, \sigma \sigma^*, y)\) and uniqueness in law up to explosion holds for the SDE \eqref{eq: Appendix SDE} if and only if the MP \((\mu, \sigma \sigma^*, y)\) has at most one solution up to explosion.
\end{theorem}
\begin{proof}
	The proof is similar as in the non-explosive case, see, for instance, the Corollaries 5.4.8 and 5.4.9 in \cite{KaraShre}. However, for the implication \(\Longleftarrow\) an additional extension argument is necessary. We sketch it. In the following the underlying filtered probability space is \((\mathbb{W}, \mathcal{W}^P, \mathbf{W}^P, P)\), where \(P\) solves the MP \((\mu, \sigma \sigma^*, y)\) up to explosion, \(\mathcal{W}^P\) is the \(P\)-completion of \(\mathcal{W}\) and \(\mathbf{W}^P\) is the \(P\)-augmentation of \(\mathbf{W}\).
	For all \(n \in \mathbb{N}\) the process
	\begin{align}\label{eq: Xc decomposition}
	X^c_{t \wedge \gamma_n} \triangleq X_{t \wedge \gamma_n} - y-  \int_0^{t \wedge \gamma_n} \mu(s, X_s) \dd s,\quad t \in [0, \infty),
	\end{align}
	is a continuous local martingale with quadratic variation process \begin{align}\left(\int_0^{t \wedge \gamma_n} \sigma(s, X_s)\sigma^* (s, X_s)\dd s\right)_{t \in [0, \infty)},\end{align}
	see Lemma II.67.10 in \cite{RW1}.
	Let \(\sigma^{-1}\) be the generalized inverse (sometimes called Moore-Penrose inverse) of \(\sigma\) and set \(\theta \triangleq \textup{Id} - \sigma^{-1} \sigma\). 
	Further, let \(W = (W_t)_{t \in [0, \infty)}\) be a \(d\)-dimensional Brownian motion defined on an extension of \((\mathbb{W}, \mathcal{W}^P, \mathbf{W}^P, P)\), which is independent of \(X\). Set 
	\begin{equation}\begin{split}
	B_t \triangleq \begin{cases} \int_0^t \sigma^{-1} (s, X_s) \dd X^c_s + \int_0^t \theta (s, X_s) \dd W_s,& t \leq \gamma_n,\\
	\liminf_{r \nearrow \gamma_I} \left(\int_0^r \sigma^{-1} (s, X_s) \dd X^c_s + \int_0^r \theta (s, X_s) \dd W_s\right),&t \geq \gamma_I.
	\end{cases}
	\end{split}
	\end{equation}
	Using the independence of \(X\) and \(W\) and the facts that \((\sigma^{-1} \sigma)^* = \sigma^{-1} \sigma\) and \(\sigma \sigma^{-1} \sigma = \sigma\), we obtain for all \(n \in \mathbb{N}\) and \(t \in [0, \infty)\)
	\begin{align}
	\lle B \rre_{t \wedge \gamma_n} = t \wedge \gamma_n \ \textup{Id}.
	\end{align}
	Thus, by Corollary 5.10 in \cite{J79}, \(B\) is a local martingale with \(\lle B\rre_t = t \wedge \gamma_I\ \textup{Id}\) for \(t \in [0, \infty)\).
	Moreover, by a theorem of Knight, see the Theorems V.1.9 and V.1.10 in \cite{RY}, possibly on a further extension of \((\mathbb{W}, \mathcal{W}^P, \mathbf{W}^P, P)\), there exists a \(d\)-dimensional Brownian motion \((\widehat{B}_t)_{t \in [0, \infty)}\) such that \(\widehat{B}_t = B_t\) for \(t < \gamma_I\). Using the independence of \(X\) and \(W\) and the fact that \(\sigma \sigma^{-1} \sigma = \sigma\), we obtain for all \(n \in \mathbb{N}\) and \(t \in [0, \infty)\)
	\begin{equation}\begin{split}
	\left\la\hspace{-0.18cm}\left\la X^c - \int_0^\cdot \sigma (s, X_s) \dd \widehat{B}_s \right\ra\hspace{-0.18cm}\right\ra_{t \wedge \gamma_n} = 0.
	\end{split} 
	\end{equation}
	In view of \eqref{eq: Xc decomposition}, we found a weak solution up to explosion to the SDE \eqref{eq: Appendix SDE}. 
\end{proof}
In view of the previous result, the following results hold for both MPs and SDEs up to explosion. We only state them for MPs.

The following proposition can be proven similar to the non-explosive case, see Exercise 6.7.4 in \cite{SV} and Theorem III.2.40 in \cite{JS}.
\begin{proposition}\label{prop: loc uni}
	Assume that \(\mu\) and \(\sigma\) are independent of time and locally bounded, and that the MP \((\mu, \sigma \sigma^*, x)\) has a unique solution up to explosion for all \(x \in I\). Furthermore, let \(\tau\) be a \((\mathcal{W}_t)_{t \in [0, \infty)}\)-stopping time. If \(P\) solves the MP \((\mu, \sigma \sigma^*, y)\) up to explosion and \(Q\) solves the stopped MP \((\mu, \sigma \sigma^*, y; \tau)\), defined as in Definition \ref{def: stopped MP} with \((\Omega, \mathcal{F})\) replaced by \((\mathbb{W}, \mathcal{W})\) and \([0, T]\) replaced by \([0, \infty)\), then \(P = Q\) on \(\sigma (X_{t \wedge \tau}, t \in [0, \infty))\).
\end{proposition}

Next, we recall Feller's test for explosion. 
\begin{theorem}\label{theo: Feller test}
	Suppose that \(I = (l, r)\) for \(- \infty \leq l < r \leq + \infty\) and let \(\mu\) and \(\sigma\) be independent of time and locally bounded such that \(\sigma\) is locally bounded away from zero. For \(c \in I\), set 
	\begin{align}
	v_c (x) \triangleq \int_c^x \exp \left( - 2\int_c^z \frac{\mu(u)}{\sigma^2(u)}\dd u \right) \int_c^z \frac{\exp \big( 2\int_c^y \frac{\mu(u)}{\sigma^2(u)} \dd u \big)}{\sigma^2(y)}\dd y \dd z.
	\end{align}
	For \(\omega \in \mathbb{W}\), set 
	\begin{equation}
	\begin{split}
	\tau_l (\omega) &\triangleq \lim_{x \searrow l} \inf(t \in [0,\infty)\colon \omega(t) = x),\\ \tau_r (\omega) &\triangleq \lim_{x \nearrow r}  \inf(t \in [0,\infty)\colon \omega(t) = x).
	\end{split}
	\end{equation}
	Then, for all initial values \(y \in I\) the MP \((\mu, \sigma^2, y)\) has a unique solution \(P\) up to explosion and \(P\) is explosive if and only if either \(\lim_{x \nearrow l} v_c(x) < \infty\) or \(\lim_{x \searrow l} v_c(x) < \infty\).
	Moreover, if \(\sigma\) is continuous the following holds:
	\begin{enumerate}
		\item[\textup{(i)}] \(P(\tau_l< \infty) > 0\) if and only if \(\lim_{x \searrow l} v_c(x) < \infty\).
		\item[\textup{(ii)}] \(P (\tau_r < \infty) > 0\) if and only if \(\lim_{x \nearrow r} v_c(x) < \infty\).
	\end{enumerate}
\end{theorem}
\begin{proof}
	First, similar arguments as used in the proof of Theorem 5.5.15 in \cite{KaraShre} yield the existence of \(P\). Now, by the Feller's test for explosion, see Theorem 5.5.29 in \cite{KaraShre}, \(P\) is explosive, i.e. \(P(\gamma_I< \infty) > 0\), if and only if \(\lim_{x \searrow l} v_c (x) < \infty\) or \(\lim_{x \nearrow r} v_c(x) < \infty\). 
	For (i) and (ii) see the proof of Theorem 5.1.5 in \cite{pinsky1995positive}.
\end{proof}
\begin{remark}
	By Problem 5.5.28 in \cite{KaraShre}, the conditions \(\lim_{x \searrow l} v_c(x) < \infty\) and \(\lim_{x \nearrow r} v_c(x) < \infty\) are independent of the choice of \(c \in I\).
\end{remark}
Finally, we observe that explosion happens arbitrarily fast with positive probability. For a proof we refer to Theorem 4.8 in \cite{Karatzas2016}.
\begin{proposition}\label{prop: fast explosion}
	Suppose that \(I = (l, r)\) for \(- \infty \leq l < r \leq + \infty\) and let \(\mu\) and \(\sigma\) be independent of time and locally bounded such that \(\sigma\) is locally bounded away from zero. If \(P\) is a solution up to explosion to the MP \((\mu, \sigma^2, y)\), then
	\(P (\gamma_I < \infty) > 0\) implies \(P (\gamma_I < \epsilon) > 0\) for all \(\epsilon > 0\).
\end{proposition}
\begin{remark}\label{rem: explosion}
	In the case where \(I \triangleq \mathbb{R}\), \(\mu\) and \(\sigma\) are independent of time and locally bounded such that \(\sigma\) is locally bounded away from zero, the Feller test also implies non-existence of a solution to the MP \((\mu, \sigma^2, y)\) as defined in Definition \ref{def: MP} with \(d = 1\). If either \(\lim_{x \searrow - \infty} v_c(x) < \infty\) or \(\lim_{x \nearrow + \infty} v_c(x) < \infty\), then, due to Theorem \ref{theo: Feller test}, the MP \((\mu, \sigma^2, y)\) has an explosive solution \(Q\). If \(P\) solves the MP \((\mu, \sigma^2, y)\) on the finite time interval \([0, T]\), then, extended in the obvious manner to \((\mathbb{W}, \mathcal{W})\), \(P\) also solves the stopped MP \((\mu, \sigma^2, y; T)\), defined as in Definition \ref{def: stopped MP} with \((\Omega, \mathcal{F})\) replaced by \((\mathbb{W}, \mathcal{W})\) and \([0, T]\) replaced by \([0, \infty)\). But then, due to the Propositions \ref{prop: loc uni} and \ref{prop: fast explosion} and Theorem \ref{theo: Feller test}, \(P(\gamma_\mathbb{R}  < T) = Q(\gamma_\mathbb{R} < T) > 0\). This contradiction yields that the MP \((\mu, \sigma^2, y)\) has no solution on the finite time interval \([0, T]\).
\end{remark}

\section{Time-Changed Stochastic Integrals}\label{Appendix: Time Change}
Take a filtered probability space \((\Sigma, \mathcal{A}, \mathbf{A}, P)\) with right-continuous filtration \(\mathbf{A} =(\mathcal{A}_t)_{t \in [0, T]}\).
Let \(C = (C_t)_{t \in [0, T]}\) be a continuous strictly increasing \(\mathbf{A}\)-adapted process with \(C_0 = 0\) and \(C_t \geq t\) for all \(t \in [0, T]\).
For each \(t \in [0, T]\) we set 
\begin{align}
\tau_t \triangleq \inf(s \in [0, T] \colon C_s \geq t),
\end{align}
which is an \(\mathbf{A}\)-stopping time such that \(\tau_t \leq t\). %In particular, \(\tau = (\tau_t)_{t \in [0, T]}\) is the left continuous inverse of \(C\). 
Since \(C\) is continuous and strictly increasing, \(\tau = (\tau_t)_{t \in [0, T]}\) is continuous and strictly increasing.
We can set \(\mathbf{A}_\tau \triangleq (\mathcal{A}_{\tau_t})_{t \in [0, T]}\), which is a right-continuous filtration on \((\Sigma, \mathcal{A})\). If \(Y = (Y_t)_{t \in [0, T]}\) is an \(\mathbf{A}\)-adapted process, then \((Y_{\tau_t})_{t \in [0, T]}\) is an \(\mathbf{A}_\tau\)-adapted process. For a proof of the next proposition see Proposition V.1.4 in \cite{RY}.
\begin{proposition}
	Let \(Y = (Y_t)_{t \in [0, T]}\) be an \(\mathbf{A}\)-adapted continuous process of finite variation and \(H = (H_t)_{t \in [0, T]}\) be an \(\mathbf{A}\)-progressively measurable process. Then,
	\begin{align}
	\int_0^{\tau_t} H_s\dd Y_s = \int_0^t H_{\tau_s} \dd Y_{\tau_s}.
	\end{align}
\end{proposition}
For a proof of the following proposition see Proposition V.1.5 in \cite{RY}.
\begin{proposition}\label{prop: TC Mg}
	If \(Y = (Y_t)_{t \in [0, T]}\) is an \(\mathbb{R}^d\)-valued local \((\mathbf{A}, P)\)-martingale, then \(Y_\tau= (Y_{\tau_t})_{t \in [0, T]}\) is a local \((\mathbf{A}_\tau, P)\)-martingale. Moreover, \(P\)-a.s.
	\(\lle Y\rre_{\tau_t} = \lle Y_\tau\rre_t\) for all \(t \in [0, T]\) and for all \(\mathbb{R}^d\)-valued \(\mathbf{A}\)-predictable processes \(H = (H_t)_{t \in [0, T]}\) with \(P\)-a.s. \(\int_0^T \la c_s H_s, H_s\ra \dd A_s < \infty\), where \(\lle Y\rre_t = \int_0^t c_s \dd A_s\), we have \(P\)-a.s. \(\int_0^t \la c_{\tau_s} H_{\tau_s}, H_{\tau_s}\ra \dd A_{\tau_s} < \infty\) and 
	\begin{align}
	\int_0^{\tau_t} H_s \dd Y_s = \int_0^t H_{\tau_s} \dd Y_{\tau_s},\quad t \in [0, T].
	\end{align}
\end{proposition}
\section{Lemmata by Le Gall, Bihari and Gronwall}\label{Appendix: Le Gall und co}
For an \(\mathbb{R}\)-valued continuous semimartingale \(Y = (Y_t)_{t \in [0, \infty)}\) denote the local time in the origin by \((L^0(Y)_t)_{t \in [0, \infty)}\). The following lemma is due to \cite{LeGall1983}. 
\begin{lemma}\label{lem: Le Gall}
	If there exists a Borel function \(\rho \colon (0, \infty) \to (0, \infty)\) such that \(\int_0^\epsilon \frac{1}{\rho(x)}\dd x = \infty\) for all \(\epsilon > 0\) and a.s.
	for all \(t \in [0, \infty)\)
	\begin{align}
	\int_0^t \frac{1}{\rho(Y_s)} \1_{\{Y_s > 0\}}\dd\hspace{0.05cm} \lle Y\rre_s < \infty,
	\end{align}
	then a.s. \(L^0(Y)_t = 0\) for all \(t \in [0, \infty)\).
\end{lemma}
The following lemma is due to \cite{Bihari1956}. 
\begin{lemma}\label{lem: Bihari}
	For \(a < b\) let \(u \colon [a, b] \to [0, \infty)\) and \(\rho \colon [0, \infty) \to [0, \infty)\) be continuous functions such that for all \(t \in [a, b]\)
	\begin{align}
	u(t) \leq \int_a^t \rho (u(s))\dd s,\quad t \in [a, b],
	\end{align}
	and \(\rho\) is strictly increasing with \(\rho(0) = 0\) and \(\int_0^\epsilon \frac{1}{\rho(x)}\dd x = \infty\) for all \(\epsilon > 0\),
	then \(u(t) = 0\) for all \(t \in [a, b]\).
\end{lemma}
Finally, we also state a Gronwall-type lemma. For a proof we refer to Lemma 4.13 in \cite{LS}.
\begin{lemma}\label{lem: Gronwall}
	Let \(c_0\) and \(c_1\) be non-negative constants, \(u \colon [0, T] \to [0, \infty)\) be a bounded Borel function and \(v \colon [0, T] \to [0, \infty)\) be a Borel function such that \(\int_0^T v(s)\dd s < \infty\) and 
	\begin{align}
	u (t) \leq c_0 + c_1 \int_0^t v(s) u(s) \dd s, \quad t \in [0, T].
	\end{align}
	Then, 
	\begin{align}
	u(t) \leq c_0 \exp \left( c_1 \int_0^tv(s)\dd s\right),\quad t \in [0, T].
	\end{align}
\end{lemma}
\section*{Acknowledgments}
The author thanks the referee for her/his time and effort devoted to the evaluation of 
the manuscript and for her/his very useful remarks.

%\bibliographystyle{abbrv}
%authordate1}
%\bibliography{References}

\end{document}